\documentclass[letter,11pt]{article}
\usepackage{graphicx}
\usepackage{amsfonts} 
\usepackage{amssymb}
\usepackage{amsthm}
\usepackage{amsmath}
\usepackage{geometry}
\usepackage{lastpage}
\usepackage{fancyhdr}
\usepackage{cite}
\usepackage{enumerate}
\usepackage{todonotes}
\usepackage{tcolorbox}
\tcbuselibrary{skins}
\usepackage[pdftex,unicode, plainpages = false, pdfpagelabels, 
bookmarks=false,
bookmarksopen = true,
bookmarksnumbered = true,
breaklinks = true,
linktocpage,
pagebackref,
colorlinks = true,  
linkcolor = blue,
urlcolor  = blue,
citecolor = red,
anchorcolor = green,
hyperindex = true,
hyperfigures
]{hyperref} 
\hypersetup{
    colorlinks=true, 
    linktoc=all,     
    linkcolor=blue,
    citecolor=red,  
}
\usepackage{thmtools}   
\usepackage{cleveref}   
\usepackage{tikz}
\usetikzlibrary{positioning, arrows.meta}
\usepackage{booktabs}
\usepackage{boxedminipage}

\usepackage{algorithm}
\usepackage[noend]{algpseudocode}
\algnewcommand{\IfSingleLine}[2]{\State \algorithmicif\ #1\ \algorithmicthen #2} 

\newcommand{\defproblem}[4]{%
\noindent
\begin{center}
\begin{boxedminipage}{1\columnwidth}
#1\\
\begin{tabular}{l p{0.75\columnwidth}}
    \textbf{Input}:    & #2\\
    \textbf{Parameter}: & #3\\
    \textbf{Question}: & #4
\end{tabular}
\end{boxedminipage}
\end{center}
}

\newtheorem{theorem}{Theorem}
\newtheorem{lemma}{Lemma}
\newtheorem{definition}{Definition}
\newtheorem{corollary}{Corollary}

\newcommand{\VMC}{\textsc{Vertex Multicut}}
\newcommand{\abbrVMC}{VMC}
\newcommand{\VMCcomp}{\textsc{Vertex Multicut Compression}}
\newcommand{\abbrVMCcomp}{VMC-Compression}

\newcommand{\lpprimal}[1]{\mathcal{L}_{primal}(#1)}
\newcommand{\lpdual}[1]{\mathcal{L}_{dual}(#1)}

\newcommand{\FCfull}[3]{\mathrm{FC}_{#1,#2}(#3)}
\newcommand{\FC}[1]{\mathrm{FC}_{W}(#1)}
\newcommand{\Partfull}[3]{P_{#1}(#2,#3)}
\newcommand{\Part}[1]{P_{W}({#1})}

\title{Faster Parameterized Vertex Multicut}
\author{
    Huairui Chu\thanks{University of California, Santa Barbara, USA. Email: huairuichu@ucsb.edu, daniello@ucsb.edu.}
    \and
    Yuxi Liu\thanks{University of Electronic Science and Technology of China, China. Email: yuxiliu823@gmail.com, jqpeng0@foxmail.com, kangyitian947@gmail.com, myxiao@uestc.edu.cn.}
    \and
    Daniel Lokshtanov\footnotemark[1]
    \and
    Junqiang Peng\footnotemark[2]
    \and
    Kangyi Tian\footnotemark[2]
    \and
    Mingyu Xiao\footnotemark[2]
}
\date{}

\begin{document}
\maketitle

\begin{abstract}
In the {\sc Vertex Multicut} 
problem the input consists of a graph $G$, integer $k$, and a set  
$\mathbf{T} = \{(s_1, t_1), \ldots, (s_p, t_p)\}$ of pairs of vertices of $G$.
The task is to find a set $X$ of at most $k$ vertices 
such that, for every $(s_i, t_i) \in \mathbf{T}$, there is no path from $s_i$ to $t_i$ in $G - X$.
Marx and Razgon [STOC 2011 and SICOMP 2014] and Bousquet, Daligault, and Thomass\'{e} [STOC 2011 and SICOMP 2018] independently and simultaneously gave the first algorithms for {\sc Vertex Multicut} with running time $f(k)n^{O(1)}$.
The running time of their algorithms is $2^{O(k^3)}n^{O(1)}$ and $2^{O(k^{O(1)})}n^{O(1)}$, respectively. 
As part of their result, Marx and Razgon introduce the {\em shadow removal} technique, which was subsequently applied in algorithms for several parameterized cut and separation problems. 
The shadow removal step is the only step of the algorithm of Marx and Razgon which requires $2^{O(k^3)}n^{O(1)}$ time. 
Chitnis et al. [TALG 2015] gave an improved version of the shadow removal step, which, among other results, led to a $k^{O(k^2)}n^{O(1)}$ time algorithm for {\sc Vertex Multicut}.

We give a faster algorithm for the {\sc Vertex Multicut} problem with running time $k^{O(k)}n^{O(1)}$. 
Our main technical contribution is a refined shadow removal step for vertex separation problems that only introduces an overhead of $k^{O(k)}\log n$ time. The new shadow removal step implies a $k^{O(k^2)}n^{O(1)}$ time algorithm for {\sc Directed Subset Feedback Vertex Set} and a $k^{O(k)}n^{O(1)}$ time algorithm for {\sc Directed Multiway Cut}, improving over the previously best known algorithms of Chitnis et al. [TALG 2015].
\end{abstract}

\newpage
\section{Introduction}

The minimum $s$--$t$ cut problem is a fundamental problem in graph theory and combinatorial optimization. 
This classical polynomial-time solvable problem forms the cornerstone of numerous separation and partitioning problems in graphs. 
A natural generalization arises when one aims to separate not just a single pair of terminals, but several terminals or terminal pairs simultaneously. 
If the goal is to separate all vertices in a given vertex set $T$ of size $p$ from each other by deleting at most $k$ edges or vertices, the problem is known as the \textsc{Multiway Cut} problem (also called the \textsc{Multi-terminal Cut} problem). 
If, instead, the input specifies a collection $\mathbf{T}$ of $p$ pairs of vertices in $G$, called {\em terminal pairs}, that must each be disconnected, the problem becomes the \textsc{Multicut} problem. 
Clearly, \textsc{Multiway Cut} is a special case of \textsc{Multicut}, where every pair of terminals defines a terminal pair. 
Depending on whether the graph is directed or undirected and whether we delete edges or vertices, we obtain four variants: 
\textsc{Edge Multicut}, \textsc{Vertex Multicut}, \textsc{Directed Edge Multicut}, and \textsc{Directed Vertex Multicut}. 
All of these variants have been extensively studied in both approximation and parameterized algorithm frameworks.

In this work we focus on parameterized algorithms for the \textsc{Vertex Multicut} problem in undirected graphs parameterized by the solution size $k$.
\defproblem
{\VMC{} (\abbrVMC{})}
{A graph $G=(V,E)$, a set $\mathbf{T}$ of vertex pairs and an integer $k$.}
{$k$.}
{Determine whether there is a vertex subset $X\subseteq V$ with $|X|\le k$ such that $\forall (s,t)\in \mathbf{T}$, $s,t$ are in different components in $G-X$.
}
Our main result is an algorithm for \textsc{Vertex Multicut} with running time $O^*(k^{O(k)})$\footnote{$O^*(\cdot)$ means omitting polynomial factors in the input size.}.

\begin{theorem}[Main]\label{thm:main}
    There is an algorithm that solves \VMC{} in time $k^{O(k)}n^{O(1)}$.
\end{theorem}
Theorem~\ref{thm:main} improves over a $O^*(2^{O(k^{O(1)})})$ time algorithm due to Bousquet, Daligault, and Thomassé~\cite{DBLP:journals/siamcomp/BousquetDT18-first-FPT-ind}, a  $O^*(2^{O(k^3)})$ time algorithm of Marx and Razgon~\cite{DBLP:journals/siamcomp/MarxR14-first-FPT}, and over the current fastest $O^*(k^{O(k^2)})$ time algorithm of Chitnis et al.~\cite{DBLP:journals/talg/ChitnisCHM15-FPT-DSFVS}.

\smallskip
\noindent
\textbf{Background and Complexity.}
There is a simple polynomial time reduction from {\sc Edge Multicut} to {\sc Vertex Multicut}  that exactly preserves solution size $k$, thus our result implies the same bound for {\sc Edge Multicut}.
 {\sc Vertex Multicut} in turn reduces to {\sc Directed Vertex Multicut} by replacing undirected edges by edges in both directions.
{\sc Directed Vertex Multicut} reduces to {\sc Directed Edge Multicut} by the well-known ``vertex-splitting'' trick that replaces every vertex with two copies, one for the in-neighbors and one for the out-neighbors, and adds an edge from the in-vertex to the out-vertex. Finally  {\sc Directed Edge Multicut} reduces back to {\sc Directed Vertex Multicut} by a directed version of the reduction from {\sc Edge Multicut} to {\sc Vertex Multicut}.
Thus the pecking order of problems (from easiest to hardest) is 
$$\mbox{\sc Edge Multicut}  \leq  \mbox{\sc Vertex Multicut} \leq \mbox{{\sc Directed Edge/Vertex Multicut}}$$
%

In undirected graphs, both \textsc{Edge Multicut} and \textsc{Vertex Multicut} with only two terminal pairs can be solved in polynomial time~\cite{2-multicut-YannakaMKPCSPC1983}. 
However, both problems become NP-hard as soon as there are three or more terminal pairs~\cite{DBLP:journals/siamcomp/DahlhausJPSY94-complexity-of-multicut}. 
From the perspective of approximation algorithms, both \textsc{Edge} and \textsc{Vertex Multicut} in undirected graphs admit an $O(\log p)$-approximation via the \textsc{Multicommodity Max-Flow} approach~\cite{DBLP:journals/siamcomp/GargVY96-logk-approxi}. On the other hand, under the Unique Games Conjecture, no constant-factor approximation is possible~\cite{DBLP:journals/cc/ChawlaKKRS06-approx-harness}.
The directed versions of the problems are substantially harder: \textsc{Directed Edge/Vertex Multicut} is NP-hard even when there are only two terminal pairs~\cite{DBLP:journals/jal/GargVY04}.
When the number of terminal pairs is part of the input, no $2^{\log^{1-\epsilon} n}$-approximation exists for any $\epsilon > 0$ unless $\text{NP} \subseteq \text{ZPP}$~\cite{DBLP:journals/jacm/ChuzhoyK09-directed-inapprox}, even when the input graph is acyclic.

\smallskip
\noindent
\textbf{Parameterized Algorithms for Cut and Separation Problems.}
In the framework of parameterized algorithms~\cite{DBLP:books/sp/CyganFKLMPPS15-book-of-para, downey2013fundamentals}, the {\sc Multicut} and {\sc Multiway Cut} problems have been considered with respect to the parameters cut-size~$k$ and the number of terminal pairs~$p$. 
A systematic study was initiated by Marx~\cite{DBLP:journals/tcs/Marx06-FPT-open-1}, who gave an algorithm for {\sc Vertex Multiway Cut} with running time $O^*(2^{O(k^3)})$ and for {\sc Vertex Multicut} with running time $O^*(2^{O(k^3 + kp)})$. 
The question whether {\sc Vertex Multicut} and {\sc Edge Multicut} are fixed parameter tractable (FPT) with respect to parameter $k$, that is whether there exists an algorithm with running time $O^*(f(k))$ for some function $f$, was a well known open problem re-stated several times~\cite{DBLP:journals/tcs/Marx06-FPT-open-1,DBLP:journals/algorithmica/ChenLL09-FPT-open-3,DBLP:journals/eor/GuoHKNU08-graph-class, DBLP:journals/tcs/BodlaenderFHMPR10-app-fuzzy-1}.


Marx~\cite{DBLP:journals/tcs/Marx06-FPT-open-1} introduced the notion of {\em important separators} (see Section~\ref{sec:seps}), which have become a standard tool in parameterized algorithms~\cite{DBLP:books/sp/CyganFKLMPPS15-book-of-para}. Marx~\cite{DBLP:journals/tcs/Marx06-FPT-open-1} proved that there are at most $4^{O(k^2)}$ important $s$-$t$ separators of size at most $k$. 
This bound was subsequently improved to $4^k$ by Chen et al.~\cite{DBLP:journals/algorithmica/ChenLL09-FPT-open-3}, who used the improved bound to get a $O^*(4^{k})$ time algorithm for {\sc Vertex Multiway Cut}. Important separators were a key tool for the first fixed parameter tractable algorithms for {\sc Directed Feedback Vertex Set}~\cite{DBLP:journals/jacm/ChenLLOR08} and {\sc Almost $2$-SAT}~\cite{DBLP:journals/jcss/RazgonO09}, each of which resolved prominent open problems in parameterized complexity.
The current fastest parameterized algorithms for {\sc Vertex Multiway Cut}~\cite{DBLP:journals/toct/CyganPPW13} 
and
{\sc Edge Multiway Cut}~\cite{DBLP:journals/mst/Xiao10, DBLP:journals/ipl/CaoCF14} run in time $O^*(2^k)$ and $O^*(1.84^k)$, respectively, and use related, but different, techniques.

The question of whether there exists a fixed parameter tractable algorithm for {\sc Vertex Multicut} was resolved simultaneously and independently by Marx and Razgon~\cite{DBLP:journals/siamcomp/MarxR14-first-FPT} and by Bousquet, Daligault, and Thomassé~\cite{DBLP:journals/siamcomp/BousquetDT18-first-FPT-ind}, who gave algorithms for \textsc{Vertex Multicut} with running time $O^*(2^{O(k^3)})$ and $O^*(2^{O(k^{O(1)})})$ respectively. 

Marx and Razgon~\cite{DBLP:journals/siamcomp/MarxR14-first-FPT} introduced the powerful technique of {\em shadow removal}. 
Loosely speaking shadow removal works whenever input is a graph $G$, a set $W$ of terminals, and potentially some additional information, and the solution we seek is a vertex set $X$.
The {\em shadow} of $X$ is the set $R$ of all vertices in $G-X$ that can not reach $W$. 
Shadow removal works for all problems where adding vertices to a solution results in a solution, and furthermore for every solution $X$ and every vertex $v \in X$, if $v$ is in the shadow of $X - \{v\}$ then $X - \{v\}$ is also a solution (in fact it works in a slightly more general setting, see Section~\ref{sec:shadowremoval}).
A {\em shadowless} solution is a solution $X$ whose shadow is empty. The shadow removal technique of Marx and Razgon~\cite{DBLP:journals/siamcomp/MarxR14-first-FPT} allows one to assume, at the cost of a $O^*(2^{O(k^3)})$ factor in the running time, that the solution is shadowless. This often greatly simplifies the search for a solution. 

Shadow removal has found multiple applications in parameterized algorithms, including algorithms for \textsc{Directed Multiway Cut}~\cite{DBLP:journals/siamcomp/ChitnisHM13-FPT-dir-multiway-cut}, 
{\sc Directed Vertex Multicut} parameterized by $k+p$ in directed {\em acyclic} graphs~\cite{DBLP:journals/siamdm/KratschPPW15}, 
\textsc{Directed Subset Feedback Vertex Set}~\cite{DBLP:journals/talg/ChitnisCHM15-FPT-DSFVS},
certain graph clustering problems~\cite{DBLP:journals/iandc/LokshtanovM13} and a parameterized approximation algorithm for {\sc Directed Odd Cycle Transversal}~\cite{DBLP:conf/soda/LokshtanovR0Z20}.

The overhead of $O^*(2^{O(k^3)})$ for the shadow removal step was improved to $O^*(2^{O(k^2)})$ by Chitnis et al~\cite{DBLP:journals/talg/ChitnisCHM15-FPT-DSFVS}, who note that this improvement also directly translates into a faster parameterized algorithm for \textsc{Directed Multiway Cut}.
The improved shadow removal step of Chitnis et al.~\cite{DBLP:journals/talg/ChitnisCHM15-FPT-DSFVS} also applies to the shadow removal step in the algorithm for {\sc Vertex Multicut} of Marx and Razgon~\cite{DBLP:journals/siamcomp/MarxR14-first-FPT}, however Chitnis et al.~\cite{DBLP:journals/talg/ChitnisCHM15-FPT-DSFVS} do not mention this.
The reason for this omission is probably that the {\sc Vertex Multicut} algorithm of Marx and Razgon has another step, called {\em reduction to the bipedal case}, which, according to the analysis of 
Marx and Razgon also introduces a $2^{O(k^3)}$ factor in the running time. 
However, it turns out that this step only introduces a $k^{O(k^2)}$ factor in the running time, and that Marx and Razgon state the less tight upper bound of $2^{O(k^3)}$, probably because that is the running time of {\em their} shadow removal step.
Indeed, at the end of the proof of Lemma 2.6 of~\cite{DBLP:journals/siamcomp/MarxR14-first-FPT}, Marx and Razgon upper bound
$(3p + |W|^{3p})^{2p}$ by $2^{O((p + \log |W|))^3}$, while the tighter bound of $p^{O(p)} \cdot |W|^{O(p^2)}$ clearly applies (to make everything extra confusing, the $p$ in the paper of Marx and Razgon~\cite{DBLP:journals/siamcomp/MarxR14-first-FPT} corresponds to $k$ in our work.)
Hence the improved shadow removal step of Chitnis et al.~\cite{DBLP:journals/talg/ChitnisCHM15-FPT-DSFVS} actually improved the algorithm for {\sc Vertex Multicut} to $O^*(k^{O(k^2)})$.
Theorem~\ref{thm:main} further improves the running time of {\sc Vertex Multicut} to $O^*(k^{O(k)})$.

We remark that {\sc Edge Multicut} retains most of the complexity of the {\sc Vertex Multicut} problem. Prior to the works~\cite{DBLP:journals/siamcomp/MarxR14-first-FPT,DBLP:journals/siamcomp/BousquetDT18-first-FPT-ind}, the existence of a parameterized algorithm for {\sc Edge Multicut} was explicitly asked as an open problem~\cite{DBLP:journals/eor/GuoHKNU08-graph-class, DBLP:journals/tcs/BodlaenderFHMPR10-app-fuzzy-1}.
In particular, the parameterized algorithms for {\sc Edge Multicut} imply parameterized algorithms for an interesting variant of a clustering problem in which, for some pairs of vertices, we neither know nor care whether they belong to the same cluster or not~\cite{DBLP:journals/tcs/BodlaenderFHMPR10-app-fuzzy-1,DBLP:journals/tcs/DemaineEFI06-app-fuzzy-2,DBLP:journals/ml/BansalBC04-app-fuzzy-3}.
Marx and Razgon~\cite{DBLP:journals/siamcomp/MarxR14-first-FPT} explicitly call the parameterized complexity of {\sc Vertex Multicut} a ``{\em very challenging open question in the area of parameterized complexity}''.

Bousquet, Daligault, and Thomassé~\cite{DBLP:journals/siamcomp/BousquetDT18-first-FPT-ind} present their entire algorithm as an algorithm for {\sc Edge Multicut}, and only sketch how to convert their algorithm to an algorithm for {\sc Vertex Multiway Cut}.



\paragraph{Methods and Corollaries.} The algorithm of Theorem~\ref{thm:main} follows the same overall framework as the algorithm of Marx and Razgon~\cite{DBLP:journals/siamcomp/MarxR14-first-FPT}. However, we completely re-design the algorithms for both of the two most technically involved steps, namely the {\em shadow removal} step and the {\em reduction to the bipedal case} step.

Our main technical contribution is a new shadow removal procedure (for the more general {\em vertex} deletion problems) that only introduces a factor of $k^{O(k)}(\log n)^{O(1)}$ in the running time, instead of the $k^{O(k^2)}(\log n)^{O(1)}$ factor of Chitnis et al.~\cite{DBLP:journals/talg/ChitnisCHM15-FPT-DSFVS}. The improved shadow removal step immediately yields \footnote{To obtain these results, it suffices to apply the directed version of our shadow removal algorithm and replace the corresponding subprocedures in their frameworks accordingly.} a $O^*(k^{O(k^2)})$ time algorithm for \textsc{Directed Subset Feedback Vertex Set},
improving over $O^*(2^{O(k^3)})$ time algorithm of Chitnis et al.~\cite{DBLP:journals/talg/ChitnisCHM15-FPT-DSFVS},
and a $O^*(k^{O(k)})$ time algorithm for \textsc{Directed Multiway Cut}, improving over a $O^*(2^{O(k^2)})$ time algorithm of Chitnis et al.~\cite{DBLP:journals/talg/ChitnisCHM15-FPT-DSFVS}.


\begin{table}[t]
    \centering
    \begin{tabular}{llll}
        \toprule
        Problem Name & New Bound & Previous Bound & Remark \\
        \midrule
        {\sc Edge Multicut} & $O^*(k^{O(k)})$ & $O^*(k^{O(k^2)})$~\cite{DBLP:journals/talg/ChitnisCHM15-FPT-DSFVS} & Main\\
        {\sc Directed Subset Feedback Vertex Set} & $O^*(k^{O(k^2)})$ & $O^*(2^{O(k^3)})$~\cite{DBLP:journals/talg/ChitnisCHM15-FPT-DSFVS} & Corollary\\
        {\sc Directed Multiway Cut} & $O^*(k^{O(k)})$ & $O^*(2^{O(k^2)})$~\cite{DBLP:journals/talg/ChitnisCHM15-FPT-DSFVS} & Corollary\\
        \bottomrule
    \end{tabular}
    \caption{New results in this paper.}
    \label{tab:new_results}
\end{table}

\paragraph{Related Work.} 

We briefly mention a few additional results regarding parameterized algorithms and approximation algorithms for variants of the {\sc Multicut} and {\sc Multiway Cut} problems. 
%
%
In undirected graphs, \textsc{Edge Multiway Cut} admits a $1.3438$-approximation algorithm~\cite{DBLP:journals/mor/KargerKSTY04}, while \textsc{Vertex Multiway Cut} allows a $2$-approximation~\cite{DBLP:journals/jal/GargVY04}. 
In directed graphs, the edge and vertex versions coincide, and a $2$-approximation is also known~\cite{DBLP:journals/siamcomp/NaorZ01}. 
%
%

Unlike the undirected case, \textsc{Directed Multicut} is known not to admit an FPT algorithm parameterized by solution size $k$ unless $\rm FPT = W[1]$~\cite{DBLP:journals/siamcomp/MarxR14-first-FPT}. 
Therefore, most work focuses on the setting where the number of terminal pairs~$p$ is fixed. 
When $p \leq 3$, the problem is FPT~\cite{DBLP:journals/siamcomp/ChitnisHM13-FPT-dir-multiway-cut,DBLP:conf/soda/HatzelJLMPSS23}, although only a randomized FPT algorithm is known for the case $p=3$. For $p = 4$ the problem becomes W[1]-hard~\cite{DBLP:journals/toct/PilipczukW18a-dir-mul-four}. 

Variants of \textsc{Multicut} have also been studied on restricted graph classes. 
For example, \textsc{Edge Multicut} on trees is NP-hard~\cite{DBLP:journals/algorithmica/GargVY97-tree-np-hard}, 
but admits an FPT algorithm~\cite{DBLP:journals/networks/GuoN05-emc-tree-fpt} and a polynomial kernel~\cite{DBLP:conf/stacs/BousquetDTY09-emc-tree-kernel} when parameterized by the cut size~$k$. 
When parameterized by the treewidth~$tw$, both \textsc{Edge Multicut} and \textsc{Vertex Multicut} become FPT if the treewidth remains bounded after connecting each terminal pair~\cite{DBLP:journals/ipl/GottlobL07-multi-logical,DBLP:conf/ciac/PichlerRW10-multi-tree-decom} by an auxiliary edge. 
For bounded-degree graphs with bounded treewidth, \textsc{Edge Multicut} admits a PTAS~\cite{DBLP:journals/jal/CalinescuFR03-bounded-deg-tw}. 

Prior to designing their FPT algorithm for {\sc Vertex Multicut}~\cite{DBLP:journals/siamcomp/MarxR14-first-FPT}, motivated by a lack of progress on achieving an FPT algorithm for this problem, Marx and Razgon~\cite{DBLP:journals/ipl/MarxR09-approx-fpt} gave a factor $2$ approximation algorithm for {\sc Vertex Multicut} running in time $O^*(2^{O(k \log k)})$. Much more recently, motivated by a lack of progress for achieving a {\em faster} FPT algorithm for this problem, Lokshtanov et al.~\cite{DBLP:conf/soda/LokshtanovMRSZ21} give a $2$-approximation algorithm for {\sc Vertex Multicut} running in time $O^*(2^{O(k)})$.

Recently, the weighted versions of {\sc Vertex Multicut} and {\sc Multiway Cut} are proved FPT by Kim et al. \cite{DBLP:journals/siamdm/KimMPSW24}, using flow augmentation \cite{DBLP:journals/jacm/KimKPW25}.




\subsection{Proof Outline}
We give a brief outline of the proof of Theorem~\ref{thm:main}. We start with a summary of the algorithm of Marx and Razgon~\cite{DBLP:journals/siamcomp/MarxR14-first-FPT}. This algorithm proceeds in four steps. The first step is an application of the {\em iterative compression} technique, and allows us, at the cost of a $O^*(k^{O(k)})$ overhead in the running time, to solve a {\sc Vertex Multicut Compression} instance instead. 
Here, in addition to the instance $(G, \mathbf{T}, k)$ we are given a vertex set $W$ of size at most $k+1$, such that $W$ is a solution to $(G, \mathbf{T}, k)$, and we are looking for a solution $X$ to $(G, \mathbf{T}, k)$ which is disjoint from $W$, and simultaneously is a multiway cut for $W$. In particular no connected component of $G-X$ should contain more than one vertex of $W$.

The second step of the algorithm is shadow removal. Recall that the shadow of a solution $X$ is the set of vertices of $G$ that can't reach $W$ in $G-X$. A solution $X$ of a {\sc Vertex Multicut Compression} instance is {\em shadowless} if its shadow is empty.

Marx and Razgon design an algorithm that given an instance $(G, \mathbf{T}, W, k)$ of {\sc Vertex Multicut Compression} and outputs $2^{O(k^3)}$ instances so that $(G, \mathbf{T}, W, k)$ is a \textbf{yes}-instance if and only if one of the output instances has a shadowless solution of size at most $k$.

This is the step that was improved by Chitnis et al.~\cite {DBLP:journals/talg/ChitnisCHM15-FPT-DSFVS}, they are able to produce only $2^{O(k^2)}$ instances while maintaining the same property. 

The third step is to reduce {\sc Vertex Multicut Compression} to it self but with {\em bipedal} instances. Here an instance $(G, \mathbf{T}, W, k)$ is bipedal if every connected component $C$ of $G - W$ satisfies $|N_G(C)| \leq 2$. The algorithm  of Marx and Razgon takes as input an instance $(G, \mathbf{T}, k)$ of {\sc Vertex Multicut Compression} and produces at most $2^{O(k^3)}$ (in fact at most $k^{O(k^2)}$, as discussed above) bipedal instances such that $(G, \mathbf{T}, W, k)$ admits a shadowless solution of size at most $k$ if and only if one of the output instances is.  

The final step of the algorithm of Marx and Razgon is to find shadowless solutions to bipedal {\sc Vertex Multicut Compression} instances. Marx and Razgon show that this step can be done in time $O^*(2^{O(k)})$ by a beautiful reduction to the {\sc Almost $2$-SAT} problem. This concludes the summary of the algorithm of Marx and Razgon. 


\paragraph{Framework of Our Algorithm.}
Our algorithm for {\sc Vertex Multicut} has the same framework as the algorithm of Marx and Razgon, but implements the two bottleneck stages (namely shadow removal and reduction to the bipedal case) in  completely new ways. 

The first step of iterative compression is essentially the same. At the cost of a $O^*(k^{O(k)})$ overhead in the running time the problem reduces to solving an {\sc Vertex Multicut Compression}. 
Here, in addition to the instance $(G, \mathbf{T}, k)$ we are given a vertex set $W$ of size at most $k$, such that $W$ is a solution to $(G, \mathbf{T}, k)$, and we are looking for an vertex set solution $X \subseteq V(G)$ to $(G, \mathbf{T}, k)$ such that, additionally, $X$ is a multiway cut for $W$.

The second step of our algorithm is an improved shadow removal step that takes an instance $(G, \mathbf{T}, W, k)$ of {\sc Vertex Multicut Compression} and outputs $k^{O(k)}$ instances so that $(G, \mathbf{T}, W, k)$ is a \textbf{yes}-instance if and only if one of the output instances has a shadowless solution of size at most $k$ \footnote{Actually what we get is slightly stronger. If the original instance is a \textbf{no}-instance, then the output instances are all \textbf{no}-instances.}.

We completely re-design the reduction to bipedal instances, and obtain an algorithm with following specifications. Given an instance $(G, \mathbf{T}, W, k)$ of {\sc Vertex Multicut Compression}, we produce $k^{O(k)}$ bipedal instances, such that $(G, \mathbf{T}, W, k)$ is a \textbf{yes}-instance with a shadowless solution if and only if one of the output instances is. Thus, at the cost of a factor of $k^{O(k)}$ in the running time we may assume that the input instance is bipedal. 

In the fourth and final step we just invoke the $O^*(2^{O(k)})$ time algorithm of Marx and Razgon for finding shadowless solutions for bipedal instances of {\sc Vertex Multicut Compression}.

We now sketch how we implement the two steps where our algorithm crucially differs from the one of Marx and Razgon, namely the shadow removal step and the reduction to bipedal instances. Out of the two, the shadow removal step is the one that appears to have the widest applicability, so we present it first. 


\paragraph{Improved Shadow Removal.}
We consider the following general setting, which encompasses {\sc Vertex Multicut Compression}. Suppose that we are working with a graph $G$, which has a special set $W$ of vertices. Every subset $X \subseteq V(G) - W$ is either a solution or not. This can for example be formalized as a function $\Phi$ that takes as input a pair $(G, W)$ and outputs a subset of $2^{V(G) - W}$, where $X \in \Phi(G,W)$ means that $X$ is a solution. 
The {\em shadow} of a vertex set $X$ is the set of vertices that can't reach $W$ in $G-X$.

We will concern ourselves only with problems (i.e. functions $\Phi$) that satisfy the following two axioms: {\em (i)} if $X$ is a solution and $Y \subseteq V(G) - W$, then $X \cup Y$ is a solution. {\em (ii)} If $X \cup \{v\}$ is a solution and $v$ is in the shadow of $X$, then $X$ is a solution. 
It is not too difficult to verify that {\sc Vertex Multicut Compression} satisfies the two axioms above, indeed {\em (i)} trivially holds. For {\em (ii)} suppose that $X \cup \{v\}$ is a solution and $v$ is in the shadow of $X$, and let $P$ be a path from  $s_i, t_i$ for $(s_i, t_i) \in \mathbf{T}$ or between two distinct vertices of $W$.
Then $X \cup \{v\}$ intersects $V(P)$, since $X \cup \{v\}$ is a solution. If  $X$ intersects $V(P)$ we are done, so suppose $v \in V(P)$. Since $W$ is a solution, $W$ intersects $P$. But then $P$ connects $v$ to $W$, and thus, since $v$ is in the shadow of $X$, $V(P)$ intersects $X$ as well. So $X$ is a solution. Hence  {\sc Vertex Multicut Compression} satisfies {\em (ii)}.

The main step in the shadow removal algorithm (ours, of Marx and Razgon~\cite{DBLP:journals/siamcomp/MarxR14-first-FPT}, and of Chitnis et al.~\cite{DBLP:journals/talg/ChitnisCHM15-FPT-DSFVS}) is a randomized algorithm, which we call {\em shadow covering}, with the following property. 
Given $(G,W,k)$ the procedure outputs a vertex set $Z$. If there exists a solution $X'$ of size at most $k$ then there exists a solution $X$ of size at most $k$, such that on every call of the procedure on $(G,W,k)$, with probability at least $p$ (the {\em success probability}) the output set $Z$ is disjoint from $X$ and $Z$ contains the shadow of $X$.

Given a shadow covering procedure with success probability $p$ one can get a randomized shadow removal algorithm that produces $O(1/p)$ instances and succeeds with probability $1-c$ for any $c > 0$. 
All we need is a graph operation that allows us to remove vertices from the graph that are not allowed to be selected into the solution (since we safely can declare the vertices in $Z$ as such unselectable vertices). 
Such an operation for {\sc Vertex Multicut Compression} was given by Marx and Razgon. Essentially, all we need to do is to remove vertices $z$ in $Z$ from the graph one by one, when $z$ is deleted its neighborhood is turned into a clique. Some care must be taken to also preserve the terminal pairs in $\mathbf{T}$ whose endpoints are deleted. 
Furthermore, all known shadow removal procedures (including ours) can be de-randomized in a standard way using splitters~\cite{DBLP:conf/focs/NaorSS95} at a negligible cost to the running time and the number of output instances. Hence in this overview we will just focus on giving a shadow covering procedure. 

Marx and Razgon~\cite{DBLP:journals/siamcomp/MarxR14-first-FPT} give a shadow covering procedure with success probability $2^{-O(k^3)}$, Chitnis et al.~\cite{DBLP:journals/talg/ChitnisCHM15-FPT-DSFVS} design one with success probability  $2^{-O(k^2)}$, while we now give one with success probability $k^{-O(k)}$.


\paragraph{Improved Shadow Covering.}
Just as the previous shadow covering procedures, ours relies on {\em important separators}. We give a very brief introduction to important separators tailored to our particular setting. For actual definitions, see Section~\ref{sec:prelim}.
We consider the setting where we are given an instance $(G, W, k)$ of a problem which satisfies the axioms {\em (i)} and {\em (ii)} of shadow removal. 

Let $v$ be a vertex of $G$. A set $X$ disjoint from $v$ and $W$ is a $v-W$ {\em separator} if the component of $G-X$ containing $v$ does not contain any vertex of $W$. 
A $v-W$ separator $X$ is a {\em minimal} $v-W$ {\em separator} if no proper subset of $X$ is also a $v-W$ separator.
A minimal $v-W$ separator $X$ is an {\em important} $v-W$ {\em separator} if there is no $v-W$ separator $X'$ distinct from $X$ such that $|X'| \leq |X|$ and the (vertex set of the) component $C'$ of $G - X'$ containing $v$ contains the component $C$ of $G - X$ containing $v$.
The important thing to know about important separators is that, for every $G, W, v, k$ there are at most $4^k$ important $v-W$ separators of size at most $k$ in $G$~\cite{DBLP:journals/algorithmica/ChenLL09-FPT-open-3, DBLP:books/sp/CyganFKLMPPS15-book-of-para}.

The two axioms {\em (i)} and {\em (ii)} imply that, if there exists a solution $X'$ to $(G, W)$ of size at most $k$ then there exists a solution $X$ of size at most $k$ such that {\em (a)} for every vertex $v$ in the shadow of $X$, the neighborhood of the component $C$ of $G-X$ that contains $v$ is an important $v-W$ separator, {\em (b)} for every $v$ in $X$ there is a path from $v$ to $W$ disjoint from $X'-\{v\}$.
The existence of $X$ given $X'$ is a step which is done in every previous paper that does shadow removal, so we do not repeat it here (it is the now standard ``pushing'' argument for important separators). 
We give a shadow covering procedure that outputs a vertex set $Z$, such that for every solution $X$ that satisfies {\em (a)} and {\em (b)}, $Z$ is disjoint from $X$ and covers the shadow of $X$ with probability at least $k^{-2k}/e$. 

The procedure relies on the following interesting property of important separators. For every $(G, W, v, k)$ there exists a vertex set $Q$ of size at most $k$, such that for every important $v-W$ separator $S$ of size at most $k$, $Q$ has non-empty intersection with $S$.
To the best of our knowledge, this lemma was first proved by Korhonen and Lokshtanov~\cite{DBLP:conf/stoc/KorhonenL23}. 
We need a slight strengthening of this lemma, namely that for any vertex set $X$ that does not separate $v$ from $W$ there exists a vertex set $Q$ of size at most $k$, {\em disjoint from} $X$, such that for every important $v-W$ separator $S$ of size at most $k$, $Q$ has non-empty intersection with $S$.
The proof of this strengthened version is essentially the same as the proof of Korhonen and Lokshtanov~\cite{DBLP:conf/stoc/KorhonenL23}, which in turn is very similar to the proof of the upper bound~\cite{DBLP:journals/algorithmica/ChenLL09-FPT-open-3, DBLP:books/sp/CyganFKLMPPS15-book-of-para} of $4^k$ on the number of important $v$-$W$ separators of size at most $k$. 

We can now describe the shadow covering procedure: Color every vertex red with probability $1/k^2$ and blue with probability $1-1/k^2$. Let $Z$ be the set of vertices $v$ such that at least one important $v-W$ separator of size at most $k$ is completely red. 

Let $X$ be a solution of size at most $k$ that satisfies {\em (a)} and {\em (b)}. 
For every $v \in X$, we may apply the important separator hitting set lemma to $v$ and $X-\{v\}$ because $X$ satisfies {\em (b)}.
Thus, for every $v$ in $X$ let $Q_v$ be the hitting set for important $v-W$ separators of size at most $k$ disjoint from $X-\{v\}$.
We need to lower bound the probability that $Z$ is disjoint from $X$ and $Z$ contains the shadow of $X$. 
It suffices to lower bound the probability that $X$ is red (since this implies that $Z$ contains the shadow of $X$, by {\em (a)}) and that for every $v \in X$, all of $Q_v$ is blue (since this prevents every important $v-W$ separator from being completely red, implying that $Z$ is disjoint from $X$).
The probability that $X$ is red is at least $(1/k^2)^k = k^{-2k}$.
Further, the probability that $\bigcup_{v \in X} Q_v$ is blue is at least $(1-1/k^2)^{k^2} \geq 1/e$.
Hence $Z$ contains the shadow of $X$ and is disjoint from $X$ with probability at least $k^{-2k}/e$.


\paragraph{Reduction to Bipedal Instances.}

We now describe our reduction from {\sc Vertex Multicut Compression} to its bipedal instances. The reduction relies on the LP relaxation of {\sc Vertex Multiway Cut}. For a {\sc Vertex Multiway Cut} instance $(G,W,k)$, we can build a linear program where the variable set is $\{d_v:v\in V(G)\setminus W\}$ and the constraints are $\sum_{v\in V(P)\setminus W} d_v\geq 1$ for every path $P$ connecting two different vertices in $W$. {\sc Vertex Multiway Cut} is equivalent to the integer version of this LP. In the relaxation version, we allow variables to take any non-negative real value. Let's call this LP $LP_{VMC}$ for now. It is known that $LP_{VMC}$ is half-integral~\cite{DBLP:journals/jal/GargVY04,DBLP:conf/iwpec/Guillemot08a}, i.e. there is always an optimal solution in which every variable takes value $0,1$ or $1/2$. This property is used to develop an algorithm for \textbf{Vertex Multiway Cut} running in time $O^*(2^{k})$\cite{DBLP:journals/toct/CyganPPW13}.

The optimal value of relaxation version serves as a lower bound for the integer version. Recall that given an instance $I=(G,\mathbf{T},W,k)$ of {\sc Vertex Multicut Compression}, any solution of the instance is a multiway cut for $W$. Thus, if $k<OPT(LP_{VMC})$ we can output \textbf{no} for the instance. Starting from this observation, our first branching rule is to try to pick a vertex $v\in V(G)\setminus W$ such that if we force $d_v=0$, $OPT(LP_{VMC})$ increases by at least $1/2$. Such a vertex is called a \emph{non-zero} vertex.

If there is a non-zero vertex $v$, then we can branch into two cases: the optimal solution for $I$ contains $v$ or not. If it contains $v$, we branch into the instance $(G-v, \mathbf{T}-v,W-v,k-1)$, where $\mathbf{T}-v$ means deleting every pair containing $v$ from $\mathbf{T}$; if it does not contain $v$, we branch into the instance where $v$ is forced not to be selected (can be encoded using graph torso). In the first branch, $k$ decreases by $1$ and in the second branch, $OPT(LP_{VMC})$ increases by $1/2$. Thus, $2k-OPT(LP_{VMC})$ always decreases by at least $1/2$. Indeed, we use $2k-OPT(LP_{VMC})$ as our measure.

In the case where there is no non-zero vertex, the above branching rule does not apply. Here we introduce the notion of isolating min separators. Note that this notion and the related branching ideas are also used in previous works~\cite{DBLP:journals/mst/Xiao10}.

An {\em isolating $w$-separator} is a vertex set $N(C)$ for a vertex set $C$ that contains $w$ and $C\cup N(C)$ is disjoint from $W-\{w\}$. Its cost is $|N(C)|$.
For every $w \in W$, let $mc(G, w, W \setminus \{w\})$ be the minimum cost among all isolating $w$-separators. Any isolating $w$-separator with this cost is called an isolating $w$-min-separator.

A critical observation is that, if there is no non-zero vertex, then it holds that $2OPT(LP_{VMC})=\sum_{w\in W} mc(G,w,W\setminus \{w\})$. We state this observation in Lemma~\ref{lem:2LP=mincut}. We sketch the proof of this lemma in the following.

It's easier to show that $2OPT(LP_{VMC})\leq \sum_{w\in W} mc(G,w,W\setminus \{w\})$. Fix a set of isolating $w$-min-separators $C_w$ for each $w\in W$. Let $d_v=1/2$ if $v$ is in one $C_w$, $d_v=1$ if $v$ is in two (or more) separators $C_{w_1},C_{w_2}$ and $d_v=0$ for all vertices not in any isolating min-separators. It's easy to check this is a feasible solution for $LP_{VMC}$, and the total cost is at most $0.5\sum_{w\in W} mc(G,w,W\setminus \{w\})$.

To show that $2OPT(LP_{VMC})\geq \sum_{w\in W}mc(G,w,W\setminus\{w\}$, we aim at constructing a half-integral optimal solution $\{d_v^*\}$ for $LP_{VMC}$ with a certain structure as following. For each $w\in W$, let $U_w$ be the set of vertices that is connected to $w$ by a path $P$ such that $d_v^*=0$ for all vertices in $P\setminus \{W\}$. $U_w$ is called the \emph{zero region} of $w$. We want our LP solution to satisfy that (i) $\forall w_1\neq w_2\in W,N(U_{w_1})\cap N(U_{w_2})=\emptyset$ and (ii) $\forall w\in W, \forall v\in N(U_w),d_v^* = 1/2$, while all other vertices take $0$ as their values. Suppose such a solution exists as an optimal solution. Then $OPT(LP_{VMC}) = 0.5\sum_{w\in W} |N(U_w)|$. Since $N(U_w)$ is an isolating $w$-separator, we have $\sum_{w\in W} |N(U_w)|\geq \sum_{w\in W}mc(G,w,W\setminus \{w\})$. It remains to show such a solution indeed exists.

We consider the \emph{dual} LP of $LP_{VMC}$, in which we assign each path $P$ connecting two different vertices in $W$ a fractional value $f_P\geq 0$ under the constraints that for each vertex $v\in V\setminus W, \sum_{v\in P} f_P\leq 1$. Let's pick an optimal solution $\{f_P^*\}$ such that the number of \emph{tight} vertices is minimized among all optimal solutions. A vertex $v$ is tight if and only if $\sum_{v\in P} f_P = 1$. We build $\{d_v^*\}$ based on $\{f_P^*\}$. For each $w\in W$, we set $U_w$ to be the set of vertices reachable from $w$ using an untight path $P$, i.e. all vertices are untight in $P\setminus \{w\}$. Then we assign $d_v^* = 1/2$ for every $v\in N(U_w)$ for some $w\in W$, and all other vertices $0$.

The constructed solution $\{d_v^*\}$ is feasible, and satisfies (i) and (ii), because if there is a vertex $v\in N(U_{w_1})\cap N(U_{w_2})$ for some $w_1\neq w_2\in W$, then by the well known complementary slackness property of LP, there is a path passing through $v$ connecting $w_1$ and $w_2$ with all vertices in it $0$ except for $v$. That means $v$ has to be a non-zero vertex, contradicting our assumption that there is no non-zero vertex. To show that $\{d_v^*\}$ is optimal, we argue that under the assumption that there is no non-zero vertex and the solution we pick has minimum number of tight vertices, every positive path (i.e. a path $P$ with $f_P^* >0$) passes through exactly two vertices with value $1/2$ (for details, see Lemma~\ref{lem:path-pass-two-N(U)}). By a counting argument, this implies that $\{d_v^*\}$ is optimal.  

Knowing that $2OPT(LP_{VMC})=\sum_{w\in W} mc(G,w,W\setminus \{w\})$, we can apply our second branching rule. We need to introduce the notion of farthest isolating min separators. A farthest isolating $w$-min-separator is the isolating $w$-min-separator $N(C)$ such that $|C|$ is maximized. Such a separator is unique and can be found using residual flow in polynomial time (see e.g.~\cite{DBLP:books/sp/CyganFKLMPPS15-book-of-para}). In our branching rule, we guess if there is \emph{contractible} solution for $I$. A solution $X$ is contractible
if and only if
\begin{itemize}
    \item $X$ is disjoint with the farthest isolating $w$-min-separator for each $w\in W$;
    \item and for each farthest isolating $w$-min-separator $N(C)$, every vertex $v\in N(C)$ is connected to some $W\setminus\{w\}$ in $G-X$. 
\end{itemize}
Note that, we name such a solution ``contractible'', because the corresponding instance can be made bipedal after applying some certain contractions.

If there is such a solution which is also shadowless, we keep the instance and proceed it into a bipedal instance later; if not, there has to be a farthest isolating $w$-min-separator $N(C)$ and there is a vertex $v\in N(C)$ such that $v$ is either in a shadowless solution $X$ or connected to $w$ in $G-X$. Notice that since $X$ is shadowless, if $v\notin X$, $v$ has to be connected to some $w\in W$ in $G-X$. Thus, for each farthest isolating $w$-min-separator $N(C)$ and every $v\in N(C)$, we create two branches. In the first branch, we delete $v$ and our measure decreases by at least one. In the second branch, we contract $v$ with $w$. In this branch, $OPT(LP_{VMC})$ does not increase since $v$ is a non-zero vertex. However, then $mc(G,w,W\setminus\{w\})$ increases, and $2OPT(LP_{VMC})=\sum_{w\in W} mc(G,w,W\setminus \{w\})$ holds no more. Hence we can immediately make progress using our first branching rule.

When the branching algorithm terminates, we get a set of instances that are assumed to admit a contractible shadowless solution, which we didn't apply any further branching rule on. These instances are, actually, very close to being bipedal. Let $N(\FC{w})$ be the farthest isolating $w$-min-separator for each $w\in W$ where $\FC{w}$ denotes the reachable set. We partition $V(G)$ by these separators. To be specific, for a set $R\subseteq W$, let $\Partfull{W}{G}{R}$ be the set of vertices that are in $\cap_{w\in R}\FC{w}$ but not in $\cup_{w\notin R}\FC{w}$. 

By analyzing the structural property of $\Partfull{W}{G}{R}$ under the assumption that there is a shadowless contractible solution, we can obtain $P_W(G,\emptyset)=\emptyset$. Intuitively, since there is a contractible solution, the inner boundary of $\Partfull{W}{G}{\emptyset}$, which is a subset of $\cup_{w\in W} N(\FC{w})$ is already separated from $W$ by the solution since the solution is contractible, and shadowless-ness implies they are empty (see Lemma~\ref{lem:compactible-p-empty}). 

Moreover, if there is an edge connecting $\Partfull{W}{G}{R_1}$ and $\Partfull{W}{G}{R_2}$ for some $R_1\neq R_2$, then either $|R_1|=|R_2|=1$ or (w.l.o.g. assume $|R_1|\leq |R_2|$) $|R_1|=1,|R_2|=2$ and $R_1\subseteq R_2$. This is proven by a simple observation: since $\Partfull{W}{G}{\emptyset}=\emptyset$, we have that
$\bigcup_{w\in W}N(V\setminus{\Part{\{w\}}})\subseteq \bigcup_{w\in W} N(\FC{w})$, implying that $|R_1|$ has to be of size $1$. Also because every $N(V\setminus \Part{\{w\}})$ is an isolating-$w$ separator, we have $\sum_{w\in W}|N(V\setminus{\Part{\{w\}}})|\geq \sum_{w\in W}|N(\FC{w})|$. Hence, the farthest isolating min separators are mutually disjoint, implying that $|R_2\setminus R_1|\leq 1$.
If we think of each $P_W(G,\{w\})$ as the ``terminals'', then the above tells us that every part $P_W(G,X)$ where $|X|\neq 1$ ``touches'' at most two terminals. Thus, we can contract $N(V\setminus P_W(G,w))$ to $w$. The produced instance is promised to be bipedal.

\paragraph{Paper Organization.}
The remainder of this paper is organized as follows. 
Section~\ref{sec:prelim} introduces necessary definitions, notation, and basic tools used throughout the paper. 
Section~\ref{sec:framework} outlines the overall framework of our improved FPT algorithm. 
Section~\ref{sec:shadowremoval} presents shadow removal algorithm and prove the improved bounds.
In Section~\ref{sec:branching}, we describe our branching algorithm.
We conclude in Section~\ref{sec:conclusion} with a discussion of possible extensions, open problems, and directions for future research.

\section{Preliminary}\label{sec:prelim}



\subsection{Basics}

Graphs considered in this paper are undirected by default, without self-loops and parallel edges. A graph is denoted by $G=(V,E)$, where $V=V(G)$ is its set of vertices and $E=E(G)$ is its multiset of edges. We use $n$ to denote the number of vertices in $G$ when $G$ is clear in the context.  An edge between $u$ and $v$ is written as an unordered pair ``$(u,v)$'' and is also considered as a set of two vertices, so $u\in (u,v)$ denotes $u$ is an endpoint of edge $(u,v)$. 
Let $S\subseteq V(G)$ be a vertex set. We denote by $G - S$ the resulting graph after removing all vertices in $S$ and all edges incident to vertices in $S$. For a set of vertex pairs $\mathbf{T}$, we denote by $\mathbf{T}-S$ the subset of $\mathbf{T}$ obtained by removing all vertex pairs containing vertices in $S$. Similarly, for an edge set $X\subseteq E(G)$ we denote by $G - X$ the resulting graph after removing all edges in $X$.

For a vertex set $S\subseteq V$, we define its open neighborhood as $N_G(S) = \{v\mid \exists u\in S,v\notin S,(u,v)\in E\}$. The closed neighborhood is then defined as $N_G[S] = N_G(S)\cup S$. We denote the set of edges going out of $S$ as $\partial_G (S) = \{(u,v)\mid u\in S,v\notin S,(u,v)\in E\}$. If the set $S$ is a singleton, i.e. $S=\{v\}$, we may simplify the notation to $N_G(v)$, $N_G[v]$, and $\partial_G(v)$. We often drop the subscript $G$ when the graph is clear from the context. 
We define the \emph{inner boundary} of a vertex set $S$ as $B(S)=N(V\setminus S)$.

In this paper, we use $[l,r]$ to denote the set of integers $\{l,l+1,l+2,...,r-1,r\}$, and let $[r]=[1,r]$. 
A path $P$ is defined as a sequence of vertices $(v_1,v_2,v_3,\dots,v_l)$ such that for all $i\in [l-1]$, either $v_i=v_{i+1}$ or $(v_i,v_{i+1})\in E$. We would say $P$ connects $v_1$ and $v_l$. The vertices $v_2,v_3,\dots,v_{l-1}$ are the internal vertices of $P$. For a vertex set $S\subseteq V$, we say a path is an $S$-path if all its internal vertices are in $S$. A subpath of $P$ is a continuous subsequence of $P$ (which is also a path). 

We say two vertices are connected if there is a path that connects them. A connected component is a maximal vertex set in which all pairs of vertices are connected. For two vertex sets $S, T\subseteq V$, a path connects them if it connects a vertex in $S$ and a vertex in $T$. Note that a vertex is always considered as connected to itself.

\subsection{Separators}\label{sec:seps}

Let $G$ be an undirected graph and $S,T\subseteq V(G)$ be disjoint vertex sets.
An \emph{$S-T$ separator} is vertex set $X \subseteq V(G)\setminus (S \cup T)$ such that there are no paths from $S$ to $T$ in graph $G - X$. 

Let $X$ be any minimal $S-T$ separator. Let $C$ be the set of all vertices which are connected to $S$ in $G-X$. Since $N(C)$ are not connected to $S$ in $G-X$, we have $N(C)\subseteq X$. By minimality, we have $X=N(C)$. Thus, we also use $N(C)$ to denote a minimal separator. When this notation is used, $C$ implicitly refers to the set of vertices connected to $S$ in $G-X$. 

Let $X=N(C)$ be a minimum $S-T$ separator, which is thus also a minimal one. We say $X$ is farthest with respect to $S$ if for all $C'\supset C$, it holds that $|N(C)|<|N(C')|$. The following lemma states the uniqueness of the farthest minimum separator, which can be found in Chapter 8 in the textbook of parameterized algorithms~\cite{DBLP:books/sp/CyganFKLMPPS15-book-of-para}.

\begin{lemma}\label{lem:unique-fc}
    If there exists an $S-T$ separator, then there exists a farthest $S-T$ minimum separator which is unique and can be found in $n^{O(1)}$ time.
\end{lemma}

Based on this lemma, a theory about ``important separators'' is established~\cite{DBLP:books/sp/CyganFKLMPPS15-book-of-para}.

\begin{definition}[Important Separator]
    Let $G=(V,E)$ be a graph. Fix source $S\subseteq V$ and sink $T\subseteq V$ where $S\cap T =\emptyset$. A minimal $S-T$ separator $N(C)$ is an important separator if, for any minimal $S-T$ separator $N(C')$ s.t. $C'\supset C$, it holds that $|N(C')|> |N(C)|$.
\end{definition}

It's easy to see the farthest minimum separator is an important separator. Moreover, we have the following lemma:

\begin{lemma}\label{lem:include_minsep}
    Let $N(C^*)$ be the farthest minimum $S-T$ separator. Then for any $S-T$ important separator $N(C)$, it holds that $C^*\subseteq C$.
\end{lemma}

Important separators remain important, even if we expand the source set.

\begin{lemma}\label{lem:impsep_sourceexpand}
    Let $N(C)$ be an important $S-T$ separator. Then for any $S'$ such that $S\subseteq S'\subseteq C$, $N(C)$ is also an important $S'-T$ separator.
\end{lemma}

\begin{proof}
    Observe that $N(C)$ is a minimal $S'-T$ separator, because any $S'-T$ separator is also an $S-T$ separator. Suppose that $N(C)$ is not an important $S'-T$ separator. Then there is a minimal $S'-T$ separator $N(C')$ such that $C\subset C'$ and $|N(C')|\leq |N(C)|$. Note that $N(C')$ is also an $S-T$ separator. We let $N(C'')$ be a minimal $S-T$ separator such that $N(C'') \subseteq N(C')$. Then we have that $C\subset C'\subseteq C''$ and $|N(C'')| \le |N(C')|\leq |N(C)|$. This contradicts that $N(C)$ is an important $S-T$ separator.
\end{proof}

Important separators are enumerable in FPT time.

\begin{lemma}\label{lem:imp_sep}
    For any graph $G=(V,E)$ and disjoint vertex sets $S,T\subseteq V$, there are at most $4^k$ important $S-T$ separators of size at most $k$. Moreover, they can be enumerated in time $4^k n^{O(1)}$.
\end{lemma}


\subsection{Multicut and Multiway Cut}

The concepts of vertex multicut and multiway cut are central in this paper. These are defined as follows.

\begin{definition}[Vertex Multicut]
    Given a graph $G=(V,E)$ and a set of vertex pairs $\mathbf{T}\subseteq V\times V$.  
    A set $X\subseteq V$ of vertices is a vertex multicut of $(G,\mathbf{T})$ if every pair $(s,t)\in \mathbf{T}$ are not connected in $G-X$.
\end{definition}

Based on the above definition, the goal of \VMC can also be stated as finding a vertex multicut of size at most $k$ for $(G,\mathbf{T})$ from the input.

\begin{definition}[Vertex Multiway Cut]
    Given a graph $G=(V,E)$ and a set of vertices $W\subseteq V$.  
    A set $X\subseteq V\setminus W$ of vertices disjoint from $W$ is a vertex multiway cut of $(G,W)$ if every component of $G-X$ contains at most one vertex in $W$.
\end{definition}

Throughout this paper, we do not distinguish the terms ``separator'' and ``cut''. We are always considering vertex separators (cuts).

\subsection{Linear Program Relaxation of Vertex Multiway Cut}

Let $G = (V, E)$ be an undirected graph and $W \subseteq V$ be a set of terminals. Let $\mathcal{P}(G,W)$ denote the set of all simple paths connecting two different terminals in $G$. 
Let $\lpprimal{G,W}$ denote the following linear program, in which the variable set is $\{d_v:v\in V\setminus W\}$.

\begin{align*}
\lpprimal{G,W}: \quad \text{minimize} \quad & \sum_{v \in V \setminus W} d_v \\
\text{subject to} \quad & \sum_{v \in P \setminus W} d_v \ge 1, \quad \forall P \in \mathcal{P}(G,W) \\
& d_v \ge 0, \quad \forall v \in V \setminus W
\end{align*}

Notice that if we restrict the solution to have integer values, then $\lpprimal{G,W}$ is equivalent to finding the minimum size of a vertex multiway cut of $(G,W)$. Thus, the optimal value of $\lpprimal{G,W}$ serves as a lower bound for that size.

The dual of the linear program $\lpprimal{G,W}$ is denoted by $\lpdual{G,W}$. It uses $\{f_{P}:P\in \mathcal{P}(G,W)\}$ as the variable set.

\begin{align*}
    \lpdual{G,W}: \quad\text{maximize} \quad & \sum_{P \in \mathcal{P}(G,W)} f_{P}\\
    \text{subject to} \quad & \sum_{P: v \in P} f_{P} \leq 1, \quad \forall v \in V \setminus W \\
    & f_{P} \geq 0, \quad \forall P \in \mathcal{P}(G,W)
\end{align*}

The following lemma follows from the well-known complementary slackness property of linear programs (see e.g. ~\cite{vazirani2001approximation}).

\begin{lemma}
    Let $\{d_v^*\}$ be an optimal solution to $\lpprimal{G,W}$, and let $\{f_P^*\}$ be an optimal solution to $\lpdual{G,W}$. Then 
    \[\sum_{v\in V\setminus W} d_v^* = \sum_{P\in \mathcal{P}(G,W)} f_P^*.\]
    Moreover, for all $v\in V\setminus W$, either $d_v^*=0$ or $\sum_{P:v\in P} f_P^* =1$; and for all $P\in \mathcal{P}(G,W)$, either $f_P^* = 0$ or $\sum_{v\in P\setminus W} d_v^* = 1$.
\end{lemma}

By the above lemma, the optimal solution values are the same for primal and dual linear programs. We denote this value to be $OPT_{lp}(G,W)$.

It is known that the primal linear program is half-integral~\cite{DBLP:conf/iwpec/Guillemot08a,DBLP:journals/jal/GargVY04}, i.e. there is an optimal solution where each vertex only takes value $0,1$ or $1/2$. 

\section{Framework}\label{sec:framework}

In this section, we illustrate the framework of our algorithm and prove Theorem~\ref{thm:main}. 

\subsection{The Main Algorithm}

Our main algorithm consists of the following parts:
\begin{enumerate}
    \item Iterative compression.
    \item Shadow removal.
    \item Making instances bipedal.
    \item Solving bipedal instances.
\end{enumerate}

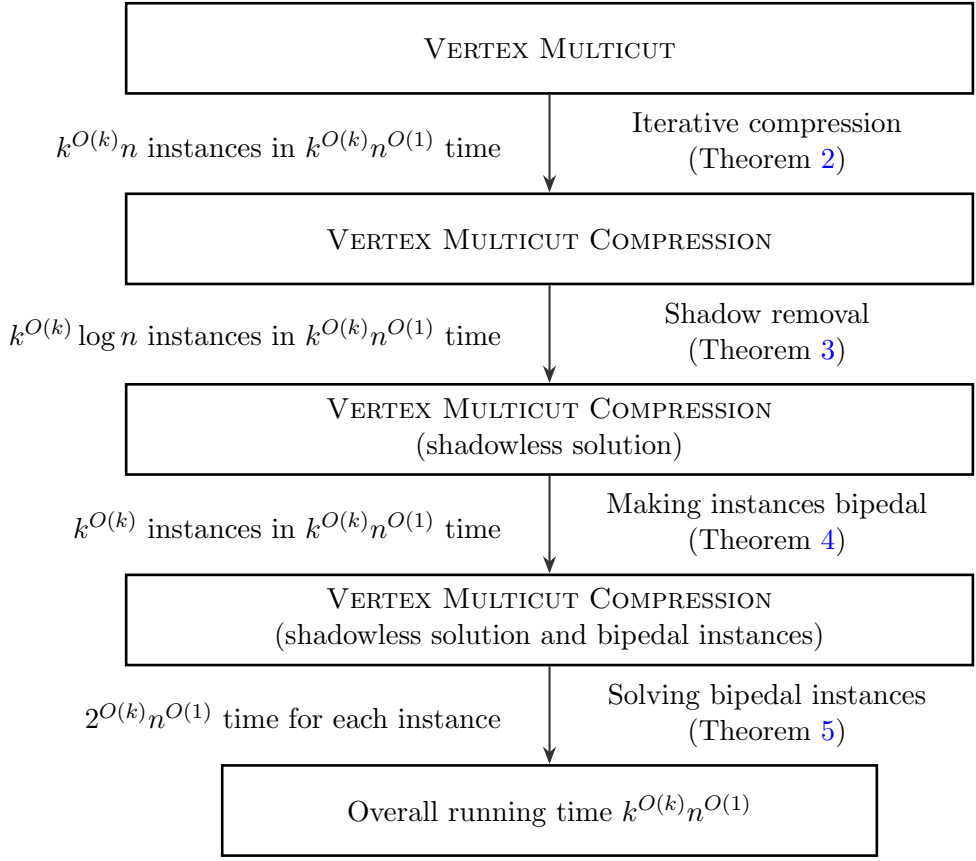
\begin{figure}[h]
    \centering
    \begin{tikzpicture}[
    node distance = 1.3cm and 1cm, 
    block/.style = {
        rectangle, 
        draw=black,  
        line width=1pt,
        text width=11cm,      
        minimum height=1.2cm, 
        text centered
    },
    final/.style = {
        rectangle, 
        draw=black, 
        line width=1pt,
        text width=8cm,
        minimum height=1.2cm, 
        text centered,
        inner xsep=10pt 
    },
    arrow/.style = {
        -Stealth, 
        thick, 
        draw=black!80
    },
    label_right/.style = {
        midway,     
        right,      
        text width=5.5cm,
        align=center  
    },
    label_left/.style = {
        midway, 
        left, 
        xshift=-5mm, 
        align=right
    }
]

    \node [block] (vmc) {\VMC{}};
    \node [block, below=of vmc] (vmc_comp) {\VMCcomp{}};
    \node [block, below=of vmc_comp] (svmc_comp) {\VMCcomp{}\\(shadowless solution)};
    \node [block, below=of svmc_comp] (bsvmc_comp) {\VMCcomp{}\\(shadowless solution and bipedal instances)};
    \node [final, below=of bsvmc_comp] (result) {Overall running time $k^{O(k)}n^{O(1)}$};

    \draw [arrow] (vmc) -- 
        node [label_right] {Iterative compression \\ (Theorem~\ref{thm:itercomp})}
        node [label_left] {$k^{O(k)}n$ instances in $k^{O(k)} n^{O(1)}$ time}
        (vmc_comp);

    \draw [arrow] (vmc_comp) -- 
        node [label_right] {Shadow removal \\ (Theorem~\ref{thm:shadowless})}
        node [label_left] {$k^{O(k)}\log n$ instances in $k^{O(k)} n^{O(1)}$ time}
        (svmc_comp);
    \draw [arrow] (svmc_comp) -- 
        node [label_right] {Making instances bipedal \\ (Theorem~\ref{thm:makebipedal})}
        node [label_left] {$k^{O(k)}$ instances in $k^{O(k)} n^{O(1)}$ time}
        (bsvmc_comp);
        
    \draw [arrow] (bsvmc_comp) -- 
    node [label_right] {Solving bipedal instances\\ (Theorem~\ref{thm:solvebipedal})}
        node [label_left] {$2^{O(k)} n^{O(1)}$ time for each instance}
    (result);

\end{tikzpicture}
    \caption{Overall framework of the algorithm for \VMC{}.}
    \label{fig:Framework}
\end{figure}

\paragraph{Part I: Iterative compression.} To start with, we use the well-know iterative compression technique, and transform the problem into a compression version.

\defproblem
{\VMCcomp{} (\abbrVMCcomp{})}
{A graph $G=(V,E)$, a set $\mathbf{T}$ of vertex pairs, a vertex subset $W\subseteq V$, and an integer $k$ such that $W$ is a vertex multiway cut of $(G,\mathbf{T})$ of size at most $k+1$.}
{$k$.}
{Determine whether there is a vertex subset $X\subseteq V$ with $|X|\le k$ such that 
\begin{itemize}
    \item $X$ is an vertex multicut of $(G,\mathbf{T})$, 
    \item $X\cap W=\emptyset$, and
    \item $X$ is a vertex multiway cut of $(G,W)$.
\end{itemize}
}

\begin{theorem}\label{thm:itercomp}
    If there is an algorithm solving \abbrVMCcomp{} in time $T(n,k)$, then there is an algorithm solving \abbrVMC{} in time $T(n,k)k^{O(k)}n^{O(1)}$.
\end{theorem}

\begin{proof}
    Let $(G, \mathbf{T}, k)$ be an instance of \VMC{}. 
    Suppose that $V(G) = \{v_1, v_2,\dots, v_n\}$, and let $V_i = \{v_1,v_2,...,v_i\}$. Let $G_i=G[V_i]$ and $\mathbf{T}_i =\{(s,t)\mid (s,t)\in \mathbf{T},s,t\in V_i\}$
    We iteratively consider the instances $I_i =(G_i, \mathbf{T}_i, k)$ in ascending order of $i$.
    Let $X_0 = \emptyset$, which is a solution of size $0\leq k$ for the empty instance $I_0$. Suppose we have a solution $X_{i-1}$ of size at most $k$ for $I_{i-1}$, we show how to find a solution for $I_i$ of size at most $k$.
    Observe that $X_{i - 1}\cup \{v_i\}$ is a solution for $I_i$ with size at most $k+1$. 
    
    We branch on all possible partitions $(W_1,\dots, W_t)$ of $X_{i-1}\cup \{v_i\}$. For a fixed such partition, we build an instance $I'=(G',\mathbf{T}',W,k-|W_1|)$. $G'$ is obtained by first removing $W_1$, then contracting each $W_j,2\leq j\leq t$ into a new vertex $w_j$ (i.e. remove $W_j$ from $V(G_i)$ and add $w_j$, and set $N(w_j)$ to be the neighbors of $W_j$ in $G_i$). $W=\{w_2,w_3,...,w_t\}$. $\mathbf{T}'$ is obtained by replacing every appearance of $v\in W_j$ by $w_j$ and remove every pair containing a vertex in $W_1$. This is equivalent to guessing that there is a solution $X$ containing $W_1$ and disjoint from $\cup_{2\leq j\leq t}W_j$, and each $W_j,2\leq j\leq t$ is connected in $G_i-X$.

    We solve each such instance using the algorithm for \abbrVMCcomp{}. If the algorithm returns a solution $X'$, we output $X_i = X'\cup W_1$. If all the instances return \textbf{no}, we return \textbf{no}. We sketch the proof of correctness. If $I_i$ is a \textbf{no}-instance, then all the constructed instances are \textbf{no}-instances, because removing a vertex and decrease the parameter by one and contracting two vertices preserves a \textbf{no}-instance. If $I_i$ has a solution $X_i$ of size $k$, then consider the partition where $W_1 = X_i\cap (X_{i-1}\cup \{v_i\})$ and each $W_j,2\leq j\leq t$ is exactly a non-empty intersection of a component in $G-X$ with $X_i\cap (X_{i-1}\cup \{v_i\})$. It's not hard to check that $X_i\setminus W_1$ is a solution for $I'$.
    
    The number of \VMCcomp{} instances we produce is bounded by $k^{O(k)}$. And for each \VMCcomp{} instance, the algorithm solving it takes $T(|V(G')|,k-|W_1|)\leq T(n,k)$ time. Thus, finding the solution (or conclude no solution) for $I_i$ takes $k^{O(k)}(T(n,k)+n^{O(1)})$ time. We need to solve $n$ instances in total, thus the overall running time is $k^{O(k)}n^{O(1)}T(n,k)$.
\end{proof}

Based on Theorem~\ref{thm:itercomp}, it suffices to give an algorithm for \abbrVMCcomp{} with running time $T(n,k)=k^{O(k)}n^{O(1)}$. We find the following definition makes describing our algorithm easier.

\begin{definition}[Partial Branching]
     A partial branching algorithm is an algorithm that takes an instance $I=(G,\mathbf{T},W,k)$ of \abbrVMCcomp{} and outputs a set of instances $\mathcal{I}$ of \abbrVMCcomp{} such that:
     \begin{itemize}
         \item For each $I'=(G',\mathbf{T}',W',k')\in \mathcal{I}$, $|I'|\leq |I|^{O(1)}$ and $k'\leq k$;
         \item If $I$ is a \textbf{no}-instance, then every instance in $\mathcal{I}$ is a \textbf{no}-instance.
     \end{itemize}
     We say that such an algorithm reduces $I$ to the set of instances $\mathcal{I}$.
 \end{definition}

\paragraph{Part II: Shadow removal.} A shadowless solution for an instance $I=(G,\mathbf{T},W,k)$ of \abbrVMCcomp{} is a solution $X$ of $I$ such that in $G-X$, every component intersects $W$. To solve the compression version, we first apply a partial branching algorithm named \verb|shadow_removal| that ensures a shadowless solution in the \textbf{yes} case.

\begin{theorem}\label{thm:shadowless}
    There is a partial branching algorithm for \abbrVMCcomp{} with running time $k^{O(k)}n^{O(1)}$, called \verb|shadow_removal|, that reduces an arbitrary instance $I$ to a set of at most $k^{O(k)}\log n$ instances $\mathcal{I}$, such that if $I$ is a \textbf{yes}-instance, then there is an instance in $\mathcal{I}$ that admits a shadowless solution.
\end{theorem}

\paragraph{Part III: Making instances bipedal.}

On an instance that can be assumed to admit a shadowless solution in the \textbf{yes} case, we apply another partial branching algorithm, named \verb|make_bipedal|. A bipedal instance $I=(G,\mathbf{T},W,k)$ of \abbrVMCcomp{} is an instance in which $G$ is bipedal with respect to $W$, i.e. $|N_G(C)\cap W|\leq 2$ holds for each component $C$ in $G-W$. The formal definition of bipedalness appears in Definition~\ref{def:bipedal}.

\begin{theorem}\label{thm:makebipedal}
    There is a partial branching algorithm for \abbrVMCcomp{} with running time $k^{O(k)}n^{O(1)}$, called \verb|make_bipedal|, that reduces an arbitrary instance $I$ to a set of at most $k^{O(k)}$ \textbf{bipedal} instances $\mathcal{I}$, such that if $I$ is a \textbf{yes}-instance with a shadowless solution, then there is an instance in $\mathcal{I}$ that also admits a shadowless solution. 
\end{theorem}

\paragraph{Part IV: Solving bipedal instances.}

The last step is to solve bipedal instances with shadowless solutions (in the \textbf{yes} case). For these instances we apply a previous result~\cite{}.

\begin{theorem}\label{thm:solvebipedal}
    There is an algorithm running in $2^{O(k)}n^{O(1)}$ time that takes a bipedal instance $I$ of \abbrVMCcomp{} as input, and
    \begin{itemize}
        \item outputs \textbf{yes} if there is a shadowless solution;
        \item outputs \textbf{no} if there is no solution.
    \end{itemize}
\end{theorem}

\paragraph{Put things together.}

\begin{proof}[Proof of Theorem~\ref{thm:main}]
    By Theorem~\ref{thm:itercomp}, it suffices to provide an algorithm solving \VMCcomp{} in $T(n,k)=k^{O(k)}n^{O(1)}$ time. Let $I$ be the input \VMCcomp{} instance. We first apply \verb|shadow_removal| on $I$. Let the output be $\mathcal{I}_1$. Then for each instance $I'\in \mathcal{I}_1$, we apply \verb|make_bipedal| on $\mathcal{I}'$. Let the union of all the output families of this step be $\mathcal{I}_2$. Finally, we apply Theorem~\ref{thm:solvebipedal} to solve every instance in $\mathcal{I}_2$, and output \textbf{yes} if and only if there is an instance in $\mathcal{I}_2$ such that the algorithm of Theorem~\ref{thm:solvebipedal} outputs \textbf{yes} on it.

    If $I$ is a \textbf{no}-instance, since \verb|shadow_removal| is a partial branching algorithm, every instance in $\mathcal{I}_1$ is a \textbf{no}-instance. And since \verb|make_bipedal| is a partial branching algorithm, every instance in $\mathcal{I}_2$ is a \textbf{no}-instance. By Theorem~\ref{thm:solvebipedal}, our algorithm outputs \textbf{no}. If $I$ is a \textbf{yes}-instance, then by Theorem~\ref{thm:shadowless}, there is an instance $I^*$ in $\mathcal{I}_1$ that admits a shadowless solution. The output of \verb|make_bipedal| on $I^*$ is a subset of $\mathcal{I}_2$. Thus by Theorem~\ref{thm:makebipedal}, there is an instance in $\mathcal{I}_2$ that admits a shadowless solution. Finally by Theorem~\ref{thm:solvebipedal}, our algorithm will output \textbf{yes}.

    The running time of our algorithm is $T(n,k)=k^{O(k)}n^{O(1)}$. Indeed, \verb|shadow_removal| takes $k^{O(k)}n^{O(1)}$ time, \verb|make_bipedal| takes $|\mathcal{I}_1|k^{O(k)}n^{O(1)}=k^{O(k)}n^{O(1)}$ time and solving the instances in $\mathcal{I_2}$ takes $|\mathcal{I_2}|2^{O(k)}n^{O(1)}=k^{O(k)}n^{O(1)}$. Hence $T(n,k)$ is the sum of cost of these three steps, which is dominated by $k^{O(k)}n^{O(1)}$.
\end{proof}

\subsection{Toolkit for Partial Branching Algorithms}

In this subsection, we introduce two operations on \abbrVMCcomp{} instances for designing partial branching algorithms, namely $\verb|torso|$ and $\verb|contract|$. Both operations never turn a \textbf{no}-instance into a \textbf{yes}-instance.

\paragraph{Graph torso.} Let $I=(G,\mathbf{T},W,k)$ be an instance of \abbrVMCcomp{}, let $Z\subseteq V(G)\setminus W$ be a set of vertices disjoint from $W$. We define $\verb|torso|(I,Z)$ to produce a new instance $I'=(G',\mathbf{T'},W',k')$, in which
\begin{itemize}
    \item $G' = (V\setminus Z, \{(u,v)\mid \text{$u,v$ are connected by a $Z$-path in $G$}\})$;
    \item $\mathbf{T}'$ is obtained by the following.
        For each vertex $v \in V$, let $\phi(v)$ be a set of vertices in $V \setminus Z$ that are connected to $v$ by a $Z$-path if $v \in Z$, and let $\phi(v) = \{v\}$ if $v \notin Z$.
        Then we set $\mathbf{T}' = \{(s',t')\mid (s',t')\in \phi(s)\times \phi(t) \text{~for some~} (s,t)\in \mathbf{T}\}$.
    \item $W'=W$.
    \item $k'=k$.
\end{itemize}

\begin{lemma}\label{lem:torso}
    Let $I=(G,\mathbf{T},W,k)$ be an instance of \abbrVMCcomp{}. $Z\subseteq V$ is a vertex set disjoint from $W$. 
    Let $I'=(G',\mathbf{T}',W',k')$ be the result of $\verb|torso|(I,Z)$,
    we have that,
    \begin{itemize}
        \item if $I$ is a \textbf{no}-instance then $I'$ is a \textbf{no}-instance;
        \item if there is a solution $X$ of size at most $k$ to $I$ such that $X\cap Z = \emptyset$, then $X$ is also a solution of $I'$.
    \end{itemize}
\end{lemma}

\begin{proof}
    Suppose $I$ is a \textbf{no}-instance. We assume $I'$ is a \textbf{yes}-instance and admits a solution $X$. We show that $X$ is also a solution to $I$. Let $P$ be a path in $G$ connecting two vertices in $W$. By replacing $Z$-paths that are subpaths of $P$ by edges in $E^+$, we get a path $P'$ in $G'$ connecting two vertices in $W$. Since $P'$ is hit by $X$ and vertices in $P'$ are all in $P$, $P$ is also hit by $X$. Now let $P$ in $G$ connects a pair $(s,t)\in \mathbf{T}$. We recall that $W$ is a vertex multicut of $(G,\textbf{T})$. Then $P$ must contain some vertex $w \notin Z$. Thus we can always find a vertex $v_s \notin Z$ closest (in the subgraph induced by $P$) to $s$ in $P$. Similarly, let $v_t$ be the vertex not in $Z$ that is closest  to $t$ in $P$. Then $(v_s,v_t)\in \mathbf{T}'$.
    By replacing $Z$-paths that are subpaths of $P$ by edges in $E^+$, we get a path $P'$ in $G'$ connecting $s'$ and $t'$. Then $P'$ is hit by $X$. Since vertices in $P'$ are all in $P$, $P$ is hit by $X$. Therefore, $I$ is a \textbf{yes}-instance with a solution $X$, which is a contradiction. Based on the above proof, $I'$ has to be a \textbf{no}-instance.
    
    Consider the other direction. Let $X$ be a solution to $I$ which satisfies that $X\cap Z=\emptyset$.
    Consider a path $P'$ in $G'$ connecting two vertices in $W$. We replace all edges in $P'$ by $Z$-paths, the path we obtained is a path $P$ connecting the two vertices in $W$ in $G$, thus hit by $X$. Since $X\cap Z=\emptyset$, $X$ also hits $P'$. Consider a path $P'$ in $G'$ connecting some $(s',t')\in \mathbf{T}'$. By definition of \verb|torso|, there is a pair $(s,t)\in \mathbf{T}$ such that $s$ and $s'$ are connected by a $Z$-path or $s=s'\notin Z$, while $t$ and $t'$ are also connected by a $Z$-path or $t=t'\notin Z$. If $s\in Z$ or $t\in Z$, we attach the corresponding $Z$-path to the end of $P'$ and replace all edges in $P'$ by $Z$-paths to obtain a path $P$ in $G$ connecting $s$ and $t$. So $P$ is hit by $X$. Since $X\cap Z = \emptyset$, $P'$ is also hit by $X$.   
\end{proof}

\paragraph{Contraction.} Let $I=(G,\mathbf{T},W,k)$ be an instance of \abbrVMCcomp{}, let $Z\subseteq V(G)\setminus W$ be a set of vertices disjoint from $W$ and let $f:Z\rightarrow W$ be a mapping. To describe the contraction operation easily, we let $f^*$ be the extension of $f$ such that $f^*(v)=f(v)$ for all $v\in Z$ and $f^*(v)=v$ for all $v\in V(G)\setminus Z$.
We define $\verb|contract|(I,f)$ to produce a new instance $I'=(G',\mathbf{T}',W',k')$, in which
\begin{itemize}
    \item $V(G')=V(G)\setminus Z$ and $E(G')=\{(f^*(u),f^*(v))\mid (u,v)\in E(G)\}$;
    \item $\mathbf{T}' = \{(f^*(s),f^*(t))\mid (s,t)\in \mathbf{T}'\}$;
    \item $W'=W$;
    \item $k'=k$.
\end{itemize}

\begin{lemma}\label{lem:contract}
    Let $I=(G,\mathbf{T},W,k)$ be an instance of \abbrVMCcomp{}, $Z\subseteq V$ is a vertex set disjoint from $W$ and $f:Z\rightarrow W$ be a mapping.
    Let $I'=(G',\mathbf{T}',W',k')$ be the result of $\verb|contract|(I,f)$,
    we have that,
    \begin{itemize}
        \item if $I$ is a \textbf{no}-instance then $I'$ is a \textbf{no}-instance;
        \item if there is a solution $X$ of size at most $k$ to $I$ such that $X\cap Z = \emptyset$, and $\forall v\in Z$, $v$ and $f(v)$ are in the same component in $G-X$, then $X$ is also a solution of $I'$.
    \end{itemize}
\end{lemma}

\begin{proof}
    If $I'$ has a solution $X$, we claim that $X$ is also a solution for $I$. Consider any path $P$ to be hit in $I$. We replace every vertex $v\in P\cap Z$ by $f(v)$, then the path becomes a path to be hit in $I'$, which has to intersect $X$. Thus, $P$ is also hit by $X$.

    If $I$ has a solution $X$ with the property as stated, consider a path $P$ to be hit in $I'$. For every edge $(u,v)$ in $P\setminus E(G)$, it has to be the case that $u\in W$ or $v\in W$. If both $u$ and $v$ are in $W$, by construction of $I'$, there exist $u',v'\in V\setminus Z$ such that $f^*(u')=u$, $f^*(v')=v$ and $(u',v')\in E(G)$. Moreover, $u',f^*(u)$ are connected and $v',f^*(v)$ are connected in $G-X$. Thus two vertices in $W$ are connected in $G-X$, contradicting that $X$ is a solution for $I$. Thus, we can assume $u\in W$ and $v\notin W$. Then there exists $u'$ such that $f(u')=u$ and $(u',v)\in E(G)$, and there is a path $P'$ connecting $u'$ and $u$ in $G-X$. Let $P'' = (u',v)$. Then $P'P''$ is a valid path in $G$. We can replace every edge $(u,v)\in P\setminus E(G)$ by such a subpath in $G$ and what we get will be a path to be hit in $I$. The new path is thus hit by $X$ (at some vertex not in ``$P'$'', since $P'$ is a path in $G-X$), which implies $P$ is hit by $X$. 
\end{proof}

\section{Shadow Removal}\label{sec:shadowremoval}

Given two vertex subsets $Y$ and $W$, the shadow of $Y$ with respect to $W$ 
is defined as vertices not reachable from $W$ after removing $Y$ from the graph. In many connectivity related problems~\cite{DBLP:journals/talg/ChitnisCHM15-FPT-DSFVS,DBLP:journals/siamcomp/ChitnisHM13-FPT-dir-multiway-cut,DBLP:journals/siamcomp/MarxR14-first-FPT}, it has been shown that finding a \emph{shadowless} solution, i.e., a solution that results in an empty shadow, is often much easier.
The shadow removal technique aims to eliminate the shadow of the original solution with respect to some $W$. In this section, we
explore how to improve the running time of the general shadow removal technique, and we then show how to apply shadow removal on \VMCcomp{} to prove Theorem~\ref{thm:shadowless}.

\subsection{Shadow Removal}
\begin{definition}[Shadow]
    Let $G = (V, E)$ be a graph. 
    Let $W \subseteq V$ and $Y \subseteq V$ be vertex sets such that $W \cap Y = \emptyset$. 
    The \emph{shadow} of $Y$ with respect to $W$ is defined as the set of vertices that are not connected to $W$ in the graph $G - Y$, that is,
    \[
        R_{G,W}(Y) = \{ v \in V \setminus W \mid \text{$v$ and $W$ are not connected in } G - Y \}.
    \]
    When the graph $G$ is clear from the context, we simply write $R_W(Y)$.
\end{definition}

The Shadow removal technique is applicable for the ``closest'' solutions. This means the solution $Y$ is pushed as close as possible to $W$. The concept of ``closest'' is formalized as $k$-shadow-removable in Definition~\ref{def:closestset}.

\begin{definition}[$k$-shadow-removable Set]\label{def:closestset}
    Let $G=(V,E)$ be a graph. Let $W\subseteq V$ and $Y\subseteq V$ be two disjoint vertex sets. We say that $Y$ is \emph{$k$-shadow-removable} with respect to $W$ if
    \begin{itemize}
        \item $|Y|\leq k$;
        \item for each vertex $v \in R_W(Y)$, there is an important $\{v\}-W$ separator $N(C)$ of size at most $k$ such that $N(C)\subseteq Y$;
        \item for each vertex $v\in Y$, any important $\{v\}-W$ separator $N(C)$ of size at most $k$ is not completely contained in $Y$, i.e. $|N(C) \setminus Y|\geq 1$.
    \end{itemize}
\end{definition}

Generally speaking, the shadow-removal technique generates a family \(\mathcal{Z}\) of vertex sets independently of the unknown solution \(Y\), with the guarantee that there exists at least one set \(Z \in \mathcal{Z}\) that covers the shadow of \(Y\) while remaining disjoint from \(Y\). 
For each \(Z \in \mathcal{Z}\), we construct a new graph in which the vertices in \(Z\) are not present, while the connectivity among the remaining vertices is preserved. 
Since there always exists some \(Z \in \mathcal{Z}\) that covers the shadow of the solution, the problem reduces to finding a shadowless solution in the resulting graphs. 
In the vertex-based setting, the family \(\mathcal{Z}\) can be handled analogously by applying the so-called \verb|torso| operation.

The main theorems, also the proof goal in this section, are Theorem~\ref{thm:shadow_removal} and Theorem~\ref{thm:deterministic_shadow_removal}. We first prove the randomized version in Theorem~\ref{thm:shadow_removal}, and the derandomized version in Theorem~\ref{thm:deterministic_shadow_removal} is obtained by standard tools.

\begin{theorem}\label{thm:shadow_removal}
    There is a randomized algorithm running in $k^{O(k)}n^{O(1)}$ time such that, on input a graph $G=(V,E)$ and a vertex set $W\subseteq V$ and a parameter $k$, outputs a family $\mathcal{Z}$ of size at most $k^{O(k)}$ of vertex sets disjoint from $W$ and ensures the following: For any $k$-shadow-removable vertex set $Y$ with respect to $W$,
    with probability at least $0.99$, there is a vertex set $Z\in \mathcal{Z}$ such that
    \begin{itemize}
        \item $Y\cap Z = \emptyset$;
        \item $R_W(Y)\subseteq Z$.
    \end{itemize}
\end{theorem}

\begin{theorem}\label{thm:deterministic_shadow_removal}
    There is a deterministic algorithm running in $k^{O(k)}n^{O(1)}$ time such that, on input a graph $G=(V,E)$ and a vertex set $W\subseteq V$ and a parameter $k$, outputs a family $\mathcal{Z}$ of size at most $k^{O(k)}\log n$ of vertex sets disjoint from $W$ and ensures the following: For any $k$-shadow-removable vertex set $Y$ with respect to $W$,
    there is a vertex set $Z\in \mathcal{Z}$ such that
    \begin{itemize}
        \item $Y\cap Z = \emptyset$;
        \item $R_W(Y)\subseteq Z$.
    \end{itemize}
\end{theorem}

\paragraph{Proof of the Shadow Removal Theorem}

Our improvement of the bound is based on the following lemma. This lemma states that there is a hitting set of important separators, whose size is small. Moreover, the hitting set can avoid a given fixed set.

\begin{lemma}\label{lem:hitting_impsep}
    Let $G=(V,E)$ be a graph, $S,T\subseteq V$ be disjoint vertex sets. Let $\mathcal{F}$ be the family of $S-T$ important separators of size at most $k$.  Let $Y\subseteq V$ be a vertex set so that for all $X\in \mathcal{F}$, $|X\setminus Y|\geq 1$. Then there exists a vertex set $H$ of size at most $k$ such that $H\cap Y=\emptyset$ and $H$ is a hitting set of $\mathcal{F}$, i.e. $\forall X\in \mathcal{F},|H\cap X|\geq 1$.
\end{lemma}

\begin{proof}
    Consider the following procedure. Let $N(C_0)$ be the farthest $S-T$ minimum separator, it is an important $S-T$ separator, and every $S-T$ important separator is its superset, by Lemma~\ref{lem:include_minsep}. Thus if $|N(C_0)|\geq k+1$, we set $H=\emptyset$. If $|N(C_0)|\leq k$, by the definition of $Y$, $|N(C_0) \setminus Y|\geq 1$. We mark an arbitrary vertex $v_1$ in $N(C_0) \setminus Y$. Now we let $N(C_1)$ be the farthest $C_0\cup\{v_1\}-T$ min separator. If $|N(C_1)|\leq k$, we mark an arbitrary vertex $v_2$ in $N(C_1) \setminus Y$, for the same reason as above. And then we turn to considering the farthest $C_1\cup\{v_2\}-T$ separator. We repeat this procedure to find $C_2,C_3,...,C_l$, until $|N(C_l)|\geq k+1$. By the uniqueness of farthest min separators, $1\leq |N(C_0)\setminus Y|\leq|N(C_0)|<|N(C_1)|<|N(C_2)|<...<|N(C_l)|$, thus $l\leq k$. Let $H$ be the set of all the marked vertices, thus $|H|= l\leq k$.
    
    It remains to prove for every important $S-T$ separator $N(C)$ s.t. $|N(C)|\leq k$, $|N(C)\cap H|\geq 1$. To show this, let $i$ be minimized such that $v_i\notin C$.
    If $i$ does not exist, we set $i=l+1$. Then $v_1,v_2,...,v_{i-1}\in C$.
    By Lemma~\ref{lem:include_minsep}, $C_0\subseteq C$. Since $N(C)$ is an important $S-T$ separator, and $S\subseteq C_0\cup \{v_1\}\subseteq C$, $N(C)$ is an important $C_0\cup \{v_1\}-T$ important separator by Lemma~\ref{lem:impsep_sourceexpand}. Then by Lemma~\ref{lem:include_minsep}, $C_1\subseteq C$. Then since $S\subseteq C_1\cup \{v_2\}\subseteq C$, $N(C)$ is an important $C_1\cup \{v_2\}-T$ separator, so $C_2\subseteq C$\dots Finally we get $C_{i-1} \subseteq C$ if $i \le l$.
    If $i=l+1$, we further have that $N(C)$ is an important $C_{l-1}\cup\{v_l\}-T$ cut. Since $N(C_l)$ is a farthest $C_{l-1}\cup\{v_l\}-T$ min cut, $|N(C)| \ge |N(C_l)| \ge k+1$, which is a contradiction. Thus we always have that $i \le l$. Then there always exists $v_i \in H$ such that $v_i \in N(C_{i-1})$, $v_i \notin C$ and $C_{i-1} \subseteq C$. This implies that $v_i\in N(C_{i-1})-C\subseteq N[C]-C=N(C)$.
\end{proof}

The idea of the shadow removal algorithm is to randomly color the graph with red and blue. Hopefully, $Y$, the set not known by the algorithm, is full of red vertices; and for every vertex $v\in Y$,  $\{v\}-W$ important separators (of size at most $k$) always contain at least one blue vertex.  To describe the good event, we introduce $(R,B)$-restricted colorings.

\begin{definition}
    Let $c:V\rightarrow \{\text{red},\text{blue}\}$ be a $2$-coloring. Let $R,B\subseteq V$ be disjoint sets. $c$ is said to be $(R,B)$-restricted if $\forall v\in R,c(v)=\text{red}$ and $\forall v\in B,c(v)=\text{blue}$.
\end{definition}

In the following lemma, we describe what the algorithm does when the coloring is given. We show that for every $k$-shadow-removable set $Y$, there are sets $R,B\subseteq V$ with small size such that the algorithm finds the required set $Z$ as long as the coloring is $(R,B)$-restricted, i.e. it colors the vertices in $R,B$ correctly. Notice that the algorithm does not know $Y$ but the produced vertex set $Z$ covers the shadow of $Y$.

\begin{lemma}\label{lem:shadow_coloring}
    There is an algorithm, on input a graph $G=(V,E)$, a vertex set $W\subseteq V$, a parameter $k$, and a $2$-coloring $c:\rightarrow \{\text{red},\text{blue}\}$, outputs in $4^k n^{O(1)}$ time a vertex set $Z\subseteq V$ disjoint from $W$ and ensures the following property:
    Let $Y\subseteq V$ be any $k$-shadow-removable set with respect to $W$.
    Then there exist disjoint sets $R,B\subseteq V$ with $|R|\leq k$ and $|B|\leq k^2$, so that for any $(R,B)$-restricted coloring $c$, on input $G,W,k,c$, the output $Z$ of the algorithm satisfies
    \begin{itemize}
        \item $Y\cap Z = \emptyset$;
        \item $R_W(Y)\subseteq Z$.
    \end{itemize}
\end{lemma}

\begin{proof}
    We pick $Z$ to be
    \[
    Z = \{v\in V-W\mid\exists \{v\}-W \text{~important separator~}N(C)\text{~of size at most~}k\text{~s.t.~}\forall u\in N(C),c(u)=\text{red}\}
    \]
    By Lemma~\ref{lem:imp_sep}, this can be computed in time $4^kn^{O(1)}$.
    Let $Y$ be as described in Lemma~\ref{lem:shadow_coloring}. Recall that for each vertex $v \in R_W(Y)$, there is an important $\{v\}-W$ separator $C\supseteq \{v\}$ of size at most $k$ such that $N(C)\subseteq Y$. By the construction of $Z$, if $\forall v\in Y,c(v)=\text{red}$, then every $v\in R_W(Y)$ will be picked into $Z$. So we can set $R=Y$.

    To ensure that $Y\cap Z=\emptyset$, we require that for every vertex $v\in Y$, any important $\{v\}-W$ separator $N(C)$ of size at most $k$ must contain a blue vertex.
    Note that for any $v \in Y$, $Y$ does not completely contain any $\{v\}-W$ important separator of size at most $k$. We can apply Lemma~\ref{lem:hitting_impsep} here to get a hitting set $H_v$ for each $v\in Y$ to hit all $\{v\}-W$ important separators of size at most $k$, such that $H_v$ is disjoint from $Y$ (thus also disjoint from $R$). Let $H=\bigcup_{v\in Y}H_v$. As long as $\forall u\in H, c(u)=\text{blue}$, $Y\cap Z=\emptyset$ then holds. So we can set $B = H$.

    The size of $R$ is $|R|=|Y|\leq k$. The size of $B$ is $|H|=|\bigcup_{v\in Y}H_v|\leq \sum_{v\in Y}|H_v|\leq k^2$. This completes the proof.
\end{proof}

With Lemma~\ref{lem:shadow_coloring}, we can build the complete randomized algorithm by adding a random coloring procedure in the beginning.

\begin{proof}[Proof of Theorem~\ref{thm:shadow_removal}]
    Consider the following randomized procedure to generate one vertex set $Z$. We first color every vertex in $V-W$ independently. With probability $p$, a vertex is colored red. Otherwise it is colored blue. Let the coloring we obtained be $c$. Then we apply the algorithm in Lemma~\ref{lem:shadow_coloring} on $G,W,k,c$ and get a set $Z$. 
    By Lemma~\ref{lem:shadow_coloring}, there exists $R,B$ of size at most $k$ and $k^2$ respectively so that $Z$ satisfies $Y\cap Z=\emptyset$ and $R_W(Y)\subseteq Z$ as long as $c$ is $(R,B)$-restricted. The probability of this event is at least $P = p^k(1-p)^{k^2}$. If we set $p = 1/k^2$, this probability becomes $P=k^{-O(k)}$.
    We repeat the above procedure $100P^{-1}$ times and let $\mathcal{Z}$ be the set of outcomes. The probability that there is a satisfying set $Z\in \mathcal{Z}$ is then at least $1-(1-P)^{100/P}\geq 1-e^{-100}\geq 0.99$.
\end{proof}

Shadow removal can be executed deterministically by using $(n,r,l)$-splitters. An $(n,r,l)$-splitter is a family of functions from $[n]\rightarrow [l]$ such that $\forall M\subseteq [n]$ with $|M|=r$, at least one of the functions in the family is injective on $M$. When $l\geq r^2$, it can be constructed deterministically in time $(n+r)^{O(1)}$ such that the size of the family is $O(r^6\log r\log n)$, by Naor et al.~\cite{DBLP:conf/focs/NaorSS95}.

\begin{proof}[Proof of Theorem~\ref{thm:deterministic_shadow_removal}]
    We construct a $(n,k+k^2,(k+k^2)^2)$-splitter where $n=|V|$. Let vertices in $V$ be $v_1,v_2,...,v_n$. For each function $f$ in the splitter and every function $g:[(k+k^2)^2]\rightarrow \{\text{red},\text{blue}\}$ such that $|g^{-1}(\text{red})| = k$. Let $c_{f,g}$ be the coloring such that $\forall v_i\in V, c_{f,g}(v_i) = g(f(i))$. Then we apply the algorithm in Lemma~\ref{lem:shadow_coloring} on $G,W,k,c_{f,g}$ and get a set $Z_{f,g}$. Let $\mathcal{Z}=\{Z_{f,g}\}$.
    By our construction, the size of $\mathcal{Z}$ is at most $k^{O(k)}\log n$.
    By Lemma~\ref{lem:shadow_coloring}, there exists $R,B$ of size at most $k$ and $k^2$ respectively so that a vertex set $Z_{f,g}$ satisfies $Y\cap Z_{f,g}=\emptyset$ and $R_W(Y)\subseteq Z_{f,g}$ as long as $c_{f,g}$ is $(R,B)$-restricted.
    Look at the function $f$ in the splitter such that $f$ is injective on $\{i\mid v_i\in R\cup B\}$ and the function $g$ such that $\forall v_i\in R$, $g(f(i))=\text{red}$ and $\forall v_i\in B,g(f(i)) = \text{blue}$. Such $g$ exists because $f$ is injective on $\{i\mid v_i\in R\cup B\}$. Then $c_{f,g}$ is $(R,B)$-restricted. So the corresponding $Z_{f,g}$ satisfies the requirements.
\end{proof}

\subsection{Apply Shadow Removal on \abbrVMCcomp}

To apply shadow removal technique, we show that a \textbf{yes}-instance always admits a $k$-shadow-removable solution. The idea is to pick a solution $X^*$ such that the number of vertices connect to $W$ is minimized in $G-X^*$. Intuitively, such a solution is ``closest'' to $W$, and it follows that $X^*$ contains important separators that separates vertices in its shadow and $W$.

\begin{lemma}\label{lem:pushing}
    Let $I=(G,\mathbf{T},W,k)$ be an instance of \VMCcomp{}. If $I$ is a \textbf{yes}-instance, then there is a solution $X^*$ of $I$ that is $k$-shadow-removable with respect to $W$.
\end{lemma}

\begin{proof}
    We say a solution to \textsc{Vertex Multicut Compression} is optimal if the cardinality is minimized. Let $X$ be a minimum optimal solution, which minimizes the number of vertices connected to $W$ in $G-X$ among all optimal solutions. To prove Lemma~\ref{lem:pushing}, it is sufficient to prove that $X$ is $k$-shadow-removable with respect to $W$.
    
    We first prove that for each vertex $v\in X$, any important $\{v\}-W$ separator of size at most $k$ is not completely contained in $X$.
    Let $R$ be the set of vertices that are connected to $W$ in $G-X$. Then $X\supseteq N(R)$. If there is a vertex $v\in X \setminus N[R]$, then we set $X'=X\setminus\{v\}$. We show that $X'$ is still a solution, which contradicts that $X$ is optimal. If $X'$ is not a solution, there must exist a path $P$ connecting two vertices in $W$ or a pair of vertices in $\mathbf{T}$ s.t. $P$ is disjoint from $X'$ but hit by $v$. Note that $P$ is also hit by $W$ in both cases, as $W$ is a vertex multicut for $(G,\mathbf{T})$. Thus there exists a subpath of $P$ connecting $v$ and $W$ without passing through $N(R)$, which is impossible.
    Therefore, we have that $X = N(R)$. Let $v$ be any vertex in $X$. Since $v \in N(R)$, there is always a $R$-path connecting $v$ and $W$ without passing any vertex in $N(R)$. Thus any $\{v\}-W$ separator is not completely contained in $X$, which implies the condition.

    We then prove that for each vertex $v \in R_W(X)$, there is an important $\{v\}-W$ separator of size at most $k$ that is totally contained in $X$.
    Now suppose that this condition is not satisfied, i.e., there is a vertex $v$ in $R_W(X)$, such that there is no important $\{v\}-W$ separator of size at most $k$ as a subset of $X$. Let the component containing $v$ in $G-X$ be $C$. Then $N(C)$ is a minimal $\{v\}-W$ separator. Since $N(C)\subseteq X$, $N(C)$ is not an important $\{v\}-W$ separator by our assumption. So, there is an $\{v\}-W$ minimal separator $N(C')$ such that $C\subset C'$ and $|N(C')|\leq |N(C)|$. Let $R' = R\setminus N[C']$. First we show that $N(R')$ hits every path connecting two vertices from $W$ or a pair in $\mathbf{T}$. Suppose not, there exists a path $P$ not hit by $N(R')$, which contains a vertex $u$ in $N(R)$ at the same time. Then there has to be a subpath of $P$ connecting $u \in N(R)$ and $W$ without passing through $N(R')$, which is impossible since $W\subseteq R'\subseteq R$. 
    
    We furtherly show that $|N(R')| \le |N(R)|$. First observe that $N(A \setminus B) \subseteq N(A) \cup B$ for any two vertex sets $A$ and $B$. Since $R' \cap N(C') = \emptyset$, it holds that $N(R') \cap C' = \emptyset$. Then we have that \[
        \begin{aligned}
            N(R') &=  N(R') \setminus C'\\
            &= N(R \setminus N[C']) \setminus C'\\ 
            &\subseteq (N(R ) \cup  N[C']) \setminus C'\\
            &\subseteq (N(R ) \setminus C') \cup  (N[C'] \setminus C') \\
            &= (N(R ) \setminus C') \cup  N(C')
        \end{aligned}
    \]
    Note that $C' \supseteq C$, thus $N(C) \subseteq N(C') \cup C'$, which implies $N(C) \setminus C' \subseteq N(C')$. Then we have that \[
        \begin{aligned}
            N(R') &\subseteq (N(R ) \setminus C') \cup  N(C')\\
            &\subseteq (N(R) \setminus N(C)) \cup ( N(C) \setminus C') \cup  N(C')\\
            &\subseteq (N(R) \setminus N(C)) \cup N(C') \cup  N(C')\\
            &= (N(R) \setminus N(C)) \cup N(C')\\
        \end{aligned}
    \]
    Since $N(C) \subseteq N(R)$ and $|N(C))| \ge |N(C')|$, it holds that $|N(R')| \le |(N(R) \setminus N(C)) \cup N(C')| \le |(N(R) |-  |N(C))| + |N(C')| \le  |N(R)| = |X|$.
    
    By the above argument, we obtain another optimal solution $N(R')$. Note that $R\cap N[C']\supseteq N(C)\cap C'\neq \emptyset$. So $R'=R\setminus N[C']\supset R$. The number of vertices connected to $W$ in $G-N(R')$ is then at most $|R'|<|R|$, contradicting the minimality of $X$. This completes the proof.
\end{proof} 

With Theorem~\ref{thm:deterministic_shadow_removal}, we can find a set $Z$ that covers the shadow in one branch. Moreover, $Z$ is disjoint from the solution. Thus we can apply \verb|torso| to build a partial branching algorithm. Thus, we are ready to prove Theorem~\ref{thm:shadowless}. Recall that the definition of a shadowless solution is a solution with empty shadow. We formalize this definition here before we prove the theorem.

\begin{definition}[Shadowless]\label{def:shadowless}
    A shadowless solution $X$ for an instance $I=(G,\mathbf{T},W,k)$ is a solution such that $R_{G,W}(X)=\emptyset$, i.e. in $G-X$, every component intersects $W$.
\end{definition}

\begin{proof}[Proof of Theorem~\ref{thm:shadowless}]
    To build the algorithm \verb|shadow_removal|, we first apply the algorithm in Theorem~\ref{thm:deterministic_shadow_removal} on $(G,W)$ and obtain a family $\mathcal{Z}$ of size at most $k^{O(k)}\log n$. Then for each $Z\in \mathcal{Z}$, run $\verb|torso|(I,Z)$ to obtain $I_Z=(G',\mathbf{T}',W',k')$. Let $\mathcal{I}$ be $\{I_Z\mid Z\in \mathcal{Z}\}$. Thus $|\mathcal{I}| \leq k^{O(k)}\log n$.

    By the definition of \verb|torso|, for every instance $I_Z=(G',\mathbf{T}',W,k)\in \mathcal{I}$, the size of $G'$ is $n^{O(1)}$ since $|V(G')|\leq |V(G)|$ and $k'=k$.
    By Lemma~\ref{lem:torso}, if $I$ is a \textbf{no}-instance, then every instance $I_Z\in\mathcal{I}$ is a \textbf{no}-instance. 
    If $I$ is a \textbf{yes}-instance, then there is a solution $X$ to $I$ which is $k$-shadow-removable, by Lemma~\ref{lem:pushing}. By Theorem~\ref{thm:deterministic_shadow_removal}, there is at least one $Z\in \mathcal{Z}$ such that $R_{G,W}(X)\subseteq  Z$ and $Z\cap X=\emptyset$. By Lemma~\ref{lem:torso}, $X$ is a solution of size at most $k$ for the instance $I_Z$. Moreover, we claim that $R_{G',W'}(X)=\emptyset$. Suppose not, let $v$ be in $R_{G',W'}(X)$. If there is a path $P$ in $G$ connecting $v$ and a vertex in $W$, then we replace every maximal $Z$-subpath $P'$ of $P$ with an edge connecting the endpoints of $P'$, obtaining a new path connecting $v$ and $W$ in $G'$. Thus $v$ and $W$ are disconnected in $G$. However, this implies $v\in R_{G,W}(X)\subseteq  Z$, contradicting that $v\in V(G')$.

    The running time of algorithm in Theorem~\ref{thm:shadowless} is $k^{O(k)}n^{O(1)}$ and the running time of $\verb|torso|$ is $n^{O(1)}$. So the total running time is $k^{O(k)}n^{O(1)}$.
\end{proof}

\section{Making Instances Bipedal}\label{sec:branching}
\newcommand{\Bnd}[1]{B(#1)}
\newcommand{\CC}[1]{CC_{#1}}
\newcommand{\mincut}{{\rm mincut}}
\newtheorem{step}{Step}




In this section, we prove the following theorem. We first recall and formally define bipedalness.

\begin{definition}[Bipedal]\label{def:bipedal}
    A graph $G$ is bipedal with respect to $W$, if and only if for every component $C$ of $G-W$, $|N_G(C)|\leq 2$.
    An instance $I=(G,\mathbf{T},W,k)$ of \abbrVMCcomp{} is bipedal if and only if $G$ is bipedal with respect to $W$.
\end{definition}

\begin{theorem}
    There is a partial branching algorithm for \abbrVMCcomp{} with running time $k^{O(k)}n^{O(1)}$, called \verb|make_bipedal|, that reduces an arbitrary instance $I$ to a set of at most $k^{O(k)}$ \textbf{bipedal} instances $\mathcal{I}$, such that if $I$ is a \textbf{yes}-instance with a shadowless solution, then there is an instance that also admits a shadowless solution. 
\end{theorem}


Our algorithm \verb|make_bipedal| consists of two steps. The first step is a partial branching algorithm \verb|branching| and the second step is to apply some certain ``contraction'' operations on each instance output by the first step. Our algorithm relies on the concept of the farthest isolating min cut, which is defined as follows.

\begin{definition}[Farthest Isolating Min Cut]
    Let $G = (V, E)$ be a graph and $W \subseteq V$ a vertex set. For a vertex $w \in W$, let $X$ be the farthest minimum $\{w\}-(W\setminus \{w\})$ cut. 
    We denote by $\FCfull{G}{W}{w}$ the vertices connected to $w$ in $G-X$.
    When the graph $G$ is clear from context, we may abbreviate this as $\FC{w}$.
\end{definition}

In subsection~\ref{subsec:branching} we describe the branching algorithm \verb|branching| and analyze its correctness and running time. In the analysis we rely on a critical lemma obtained by analyzing the primal and dual LPs. The lemma is left proven in Subsection~\ref{subsec:MeasureLemma}. In the last subsection we show that contracting the instances produced by \verb|branching| in a certain way produces bipedal instances.

\subsection{Branching into Instances with Contractible Solutions} \label{subsec:branching}

In this subsection, we provide our branching algorithm \verb|branching|. This algorithm is a partial branching algorithm and ensures that in the \textbf{yes} case, one of the reduced instance contains a \emph{contractible} solution. The definition of a contractible solution is the following:

\begin{definition}[Contractible Solutions]
    Let $I = (G, \mathbf{T}, W, k)$ be an instance of \VMCcomp{}. A solution $X$ to $I$ is called \textit{contractible} if the following conditions hold:
    \begin{itemize}
        \item For every vertex $v \in \bigcup_{w \in W} N(\FCfull{G}{W}{w}) \setminus W$, we have $v \notin X$.
        \item For every $w \in W$ and every vertex $v \in N(\FCfull{G}{W}{w}) \setminus W$, the vertices $v$ and $w$ lie in different connected components of $G - X$.
        %
    \end{itemize}
\end{definition}

\begin{lemma}
    There is a partial branching algorithm for \abbrVMCcomp{}, called $\verb|branching|$ that runs in time $k^{O(k)}$ and reduces an arbitrary instance $I$ to a set of at most $k^{O(k)}$ instances $\mathcal{I}$ such that if $I$ is a \textbf{yes}-instance with a shadowless solution, then there is a \textbf{yes}-instance $I'=(G',\mathbf{T}',W',k')\in \mathcal{I}$ with a shadowless solution $X$ which is also contractible.
    \label{lem:branch-into-contra}
\end{lemma}

Our branching algorithm uses $2k-OPT_{lp}(G,W)$ as the measure. We have the reduction rule and branching rules as follows.

\textbf{Reduction Rule 1.} If $2k - OPT_{lp}(G,W) <0$, then the algorithm terminates and returns nothing.

\begin{lemma}
    If $2k - OPT_{lp}(G,W) <0$, then $I$ is a \textbf{no}-instance.
    \label{lem:measure-correctness}
\end{lemma}
\begin{proof}
    It is sufficient to prove that if $I$ is a \textbf{yes}-instance, then $OPT_{lp}(G,W) \le 2k$.
    Suppose that $I$ is a \textbf{yes}-instance. Let $X$ be any solution of $I$. Thus $|X| \le k$. We can set the variable of $\lpprimal{G,W}$ as follows: $d_v=1$ if $v \in X$ and $d_v=0$ if $v \in V \setminus (W \cup X)$. Then it satisfies the constraints of $\lpprimal{G,W}$ since $X$ is a vertex multiway cut of $(G,W)$. Thus we have that $OPT_{lp}(G,W) \le |X| \le  2k$. This completes the proof.
\end{proof}

If Reduction Rule 1 cannot be applied anymore, we apply Branching Rule 1.

\begin{definition}
    A vertex $v\in V(G) \setminus W$ is called a \emph{non-zero vertex} if 
    in every optimal solution $\{d^*_{u}\}$ to $\lpprimal{G,W}$, $d^*_{v}>0$.
\end{definition}

\textbf{Branching Rule 1.} If there is a non-zero vertex, we branch into two sub-instances:
\begin{itemize}
    \item $\verb|torso|((G,\mathbf{T},W,k),\{v\})$
\end{itemize}

In the case where Branching Rule 1 does not apply, i.e., there is no non-zero vertex, we apply Branching Rule 2.
We recall that $\mathcal{I}$ is the output of the algorithm $\verb|branching|$.

\textbf{Branching Rule 2.} Put $(G,\mathbf{T},W,k)$ into $\mathcal{I}$. For each vertex $w\in W$ and each vertex $v\in N(\FCfull{G}{W}{w})\setminus W$, let $f:\{v\}\to W$ be a mapping with $f(v)=w$. Let $(G',\mathbf{T}',W',k')$ be the result of $\verb|contract|(I,f)$. For each $v$ and $w$, we branch into two sub-instances:
\begin{itemize}
    \item $(G-\{v\},\mathbf{T}-\{v\},W,k-1)$;
    \item $(G',\mathbf{T}',W',k')$.
\end{itemize}
Since $|N(\FCfull{G}{W}{w})|\leq k$, there are at most $O(k^2)$ branches.

\begin{algorithm}[!ht]
\caption{\texttt{Branching}$(G, \mathbf{T}, W, k)$}
\label{alg:branch}
\begin{algorithmic}[1]
\Statex \textbf{Input:} An instance $I = (G, \mathbf{T}, W, k)$ of \VMCcomp{}.
\Statex \textbf{Output:} A family $\mathcal{I}$ of \VMCcomp{} instances.

\State Initialize $\mathcal{I} \gets \emptyset$. \Comment{Family of generated instances.}

\If{$2k - OPT(\mathcal{L}_{P}(G,W)) < 0$}\label{line:terminate} \Comment{Termination condition (by Lemma~\ref{lem-measure-correctness}).}
    \State \Return $\mathcal{I}=\emptyset$. 
\EndIf

\If{there is a non-zero vertex $v\in V(G) \setminus W$}
    \State Let $I'=(G',\mathbf{T}',W,k)$ be the result of $\verb|torso|(I,\{v\})$
    \State $\mathcal{I} \gets \mathcal{I} \cup \texttt{Branching}(G', \mathbf{T}', W, k)$.
    \label{algline:x1}
        \Comment{torso $v$.}
    \State $\mathcal{I} \gets \mathcal{I} \cup \texttt{Branching}(G-\{v\}, \mathbf{T}-\{v\}, W, k-1)$. \label{algline:y2}
        \Comment{delete $v$.}
    \State \Return $\mathcal{I}$.
\Else \Comment{Assume that there is no one-vertex.}
\State Add $I$ to $\mathcal{I}$. \Comment{$I$ itself may have a contractible solution.}
\If{$\bigcup_{w \in W} N(\FCfull{G}{W}{w}) \setminus W \neq \emptyset$}
    \For{each $w \in W$ and $v \in N(\FCfull{G}{W}{w}) \setminus W$}
        \State Let $f:\{v\}\to W$ be a mapping with $f(v)=w$.
        \State Let $(G',\mathbf{T}',W',k')$ be the result of $\verb|contract|(I,f)$.
        \State $\mathcal{I} \gets \mathcal{I} \cup \texttt{Branching}(G', \mathbf{T}', W', k')$. \label{algline:y1}
        \Comment{contract $v$ into $w$.}
        \State $\mathcal{I} \gets \mathcal{I} \cup \texttt{Branching}(G-\{v\}, \mathbf{T}-\{v\}, W, k-1)$. \label{algline:y2}
        \Comment{delete $v$.}
    \EndFor
\EndIf
\State \Return $\mathcal{I}$.
\EndIf

\end{algorithmic}
\end{algorithm}


We start the analysis of \verb|branching| and aim at proving Lemma~\ref{lem:branch-into-contra}. The following lemma is critical in our analysis.

\begin{restatable}{lemma}{lemLPeqMincut}
\label{lem:2LP=mincut}
    If there is no non-zero vertex, then we have 
    \[
        2OPT_{lp}(G,W) = \sum_{w\in W} \mincut(G,w,W\setminus\{w\}).
    \]
\end{restatable}

We leave the proof of Lemma~\ref{lem:2LP=mincut} for Section~\ref{subsec:MeasureLemma}.

\begin{lemma}
    Let $I = (G,\mathbf{T},W,k)$ be a \textbf{yes}-instance of \abbrVMCcomp{} with a shadowless solution $X$. For any vertex $v \in V\setminus W$ and $Z\subseteq V \setminus W$, the following properties hold: 
    \begin{itemize}
        \item[1.] $X \setminus \{v\}$ is a shadowless solution to the instance $(G-\{v\},\mathbf{T}-\{v\},W,k-1)$ if $v \in X$.
        \item[2.] Let $I'=((G',\mathbf{T}',W,k)$ be the result of $\verb|torso|(I,\{v\})$. $X$ is a shadowless solution to the instance $I'$ if $v \notin X$.
        \item[3.] Let $f:Z\rightarrow W$ be a mapping. Let $I'=(G',\mathbf{T}',W,k)$ be the result of $\verb|contract|(I,f)$,
        if $X\cap Z = \emptyset$, and $\forall v\in Z$ and $v$ is in the same component as $f(v)$ in $G-X$, then $X$ is also a shadowless solution to $I'$.
        
    \end{itemize}
    \label{lem:shadow-correctness}
\end{lemma}

\begin{proof}
    For property 1, it is easy to see that $X\setminus\{v\}$ is a solution of the new instance $(G-\{v\},\mathbf{T}-\{v\},W,k-1)$. Note that $(G\setminus\{v\})-(X\setminus \{v\})=G-X$. Thus $X\setminus\{v\}$ is still a shadowless solution of the new instance.
    
    For property 2, $X$ is a solution of the instance $I'$ by Lemma~\ref{lem:torso}.
    We then prove that if there exists a path $P$ from $v_1 \in V \setminus (W \cup \{v\})$ to some $w \in W$ in the graph $G - X$, then there exists a path $P'$ from $v_1$ to $w$ in the graph $G' -X$. First we have that $V(P) \cap X=\emptyset$. If $P$ does not contain $v$, we can simply let $P'=P$. Otherwise, $P$ contains $v$. Let $v_2$ and $v_3$ be the predecessor and successor of $v$. By the definition of the torso, we know that $v_2$ are adjacent to $v_3$ in $G'$. Let $P_1$ be the sub-path of $P$ from $v_1$ to $v_2$ and $P_2$ be the sub-path of $P$ from $v_3$ to $w$. We let $P'$ be the concatenation of the path $P_1$, the edge $(v_2,v_3)$ and the path $P_2$. Then $P'$ is a path in $G'$ connecting $v_1$ and $w$.
    Since $X$ is a shadowless solution of $I$, each vertex $v_1 \in V \setminus (W \cup \{v\})$ is connected to some vertex $w \in W$ in the graph $G-X$. Then each vertex $v_1 \in V \setminus (W \cup \{v\})$ is also connected to some vertex $w \in W$ in the graph $G'-X$. Thus $X$ is a shadowless solution of $I'$.
    
    
    For property 3, $X$ is a solution of the instance by Lemma~\ref{lem:contract}.
    We then prove that if there exists a path $P$ from $v_1 \in V \setminus (W \cup Z)$ to some $w \in W$ in the graph $G - X$, then there exists a path $P'$ from $v_1$ to some $w' \in W$ in the graph $G' -X$. First we have that $V(P) \cap X=\emptyset$. If $V(P) \cap Z = \emptyset$, we can simply let $P'=P$. Otherwise, let $v$ be the closest vertex in $P$ from $v_1$ that is contained in $Z$. Let $v_2$ be the predecessor of $v$. By the definition of the contraction, we know that $v_2$ are adjacent to $w' = f(v) \in W$ in $G'$. Let $P_1$ be the sub-path of $P$ from $v_1$ to $v_2$. We let $P'$ be the concatenation of the path $P_1$ and the edge $(v_2,w')$. Then $P'$ is a path in $G'$ connecting $v_1$ and $w'$.
    Since $X$ is a shadowless solution of $I$, each vertex $v_1 \in V \setminus (W \cup Z)$ is connected to some vertex $w \in W$ in the graph $G-X$. Then each vertex $v_1 \in V \setminus (W \cup Z)$ is also connected to some vertex $w' \in W$ in the graph $G'-X$. Thus $X$ is a shadowless solution of $I'$.
\end{proof}

We then prove the correctness lemma of the branching rules.

\begin{lemma}[Correctness]
    Let $I = (G,\mathbf{T},W,k)$ be an instance of \abbrVMCcomp{} and $\mathcal{I}$ be output of the algorithm $\verb|branching|$. If $I$ is a \textbf{yes}-instance with a shadowless solution, then there is a \textbf{yes}-instance $I'=(G',\mathbf{T}',W',k')\in \mathcal{I}$ with a shadowless solution which is also contractible.
    \label{lem:branching-correctness}  
\end{lemma}

\begin{proof}
    Suppose that the input instance $I = (G, \mathbf{T}, W, k)$ admits a shadowless solution. Let $X$ be the shadowless solution of $I$. By Lemma~\ref{lem:measure-correctness}, the termination condition would not hold. We prove that if $I$ admits a shadowless solution, then we either put a \textbf{yes}-instance with a shadowless and contractible solution into $\mathcal{I}$ or branch into sub-instances that admits a shadowless solution.
    
    Suppose that there is a non-zero vertex $v \in V\setminus W$. If $v \in X$, then $X \setminus \{v\}$ is a shadowless solution of the new instance $(G - {v}, \mathbf{T}-\{v\}, W,k-1)$ by Property 1 of Lemma~\ref{lem:shadow-correctness}. Let $I' = (G',\mathbf{T}', W,k)$ be the result of $\verb|torso|(G,\{v\})$. If $v \notin X$, then $X$ is also a shadowless solution of the new instance $I'$ by Property 2 of Lemma~\ref{lem:shadow-correctness}.
    
    Otherwise, we can assume that there is no non-zero vertex $v \in V(G)\setminus W$.
    If $X$ is a contractible solution, then Lemma~\ref{lem:branching-correctness} holds since we first put $I$ into $\mathcal{I}$ in Branching Rule 2. Otherwise, by the definition of contractible solutions, one of the two following cases will occur.
    In the first case, for some vertex $v \in \bigcup_{w \in W} N(\FCfull{G}{W}{w}) \setminus W$, we have $v \in X$. By Property 1 of Lemma~\ref{lem:shadow-correctness}, $X \setminus \{v\}$ is a shadowless solution of the instance $(G-\{v\},\mathbf{T}-\{v\},W,k-1)$, which is a sub-instance of Branching Rule 2.
    In the second case, for some $w \in W$ and some vertex $v \in N(\FCfull{G}{W}{w}) \setminus W$, the vertices $v$ and $w$ lie in the same connected components of $G - X$. Let $f$ be the mapping defined in Branching Rule 2 and $(G',\mathbf{T}',W',k')$ be the result of $\verb|contract|(I,f)$. By Property 3 of Lemma~\ref{lem:shadow-correctness}, $X$ is a shadowless solution of the instance $(G',\mathbf{T}',W',k')$, which is also a sub-instance of Branching Rule 2.

Above all, we have shown that either we put an instance with a contractible and shadowless solution into the output $\mathcal{I}$ or branch into a sub-instance with a shadowless solution. By a conventional induction argument, it can be shown that at least one generated instance admits a contractible and shadowless solution if the input instance admits a shadowless solution, which completes the proof.
\end{proof}


\begin{lemma}
    Let $I=(G,\mathbf{T},W,k)$ be an instance of \abbrVMCcomp{}
    \begin{itemize}
        \item[1.] Let $v \in V \setminus W$ be any non-zero vertex and $I'=(G',\mathbf{T}',W,k)$ be the result of $\verb|torso|(I,Z)$. Then we have that $OPT_{lp}(G',W) \ge OPT_{lp}(G,W) + \frac{1}{2}$.
        \item[2.] Let $v \in V$ be any vertex disjoint from $W$ and $w$ be any vertex in $W$. Let $f:\{v\}\to W$ be a mapping such that $f(v)=w$ and $I'=(G',\mathbf{T}',W,k)$ be the result of $\verb|contract|(I,f)$. Then we have that $OPT_{lp}(G',W) \ge OPT_{lp}(G,W)$.
    \end{itemize}
    
    \label{lem:OPT-torso-contra-vertex}
\end{lemma}
\begin{proof}
    We first prove Property 1. By the half-integrality of $\lpprimal{G,W}$, it is sufficient to prove that $OPT_{lp}(G',W) > OPT_{lp}(G,W)$.
    First note that $\lpprimal{G,W}$ is equivalent to $\lpprimal{G',W}$ by setting the variable of $v$ as $0$. Thus we naturally have that $OPT_{lp}(G',W) \ge OPT_{lp}(G,W)$. 
    Suppose that $OPT_{lp}(G',W) = OPT_{lp}(G,W)$. This implies that there exists an optimal solution of $\lpprimal{G,W}$ such that $d_v=0$. However, $v$ is a non-zero vertex, which is a contradiction. Above all, we have that $OPT_{lp}(G',W) \neq OPT_{lp}(G,W)$ and thus $OPT_{lp}(G',W) > OPT_{lp}(G,W)$, which completes the proof.

    We then prove Property 2. Let $\{d_u'\}$ be any optimal solution of $\lpprimal{G',W}$. We set variables $\{d_u\}$ such that $d_v=0$ and $d_u = d_u'$ for all $u \in V \setminus \{v\}$. We then prove that $\{d_u\}$ is a valid solution of $\lpprimal{G,W}$. Suppose not. Then we can find a path $P$ of $G$ such that $\sum_{u \in V(P)} d_u <1$. If $v \notin V(P)$, $P$ is also in $G'$ and thus $\sum_{u \in V(P)} d_u' = \sum_{u \in V(P)} d_u<1$, which is a contradiction. Otherwise, we can assume that $v \in V(P)$. Then $v$ will be contracted into $w=f(v)$ in the graph $G'$. Let $P'$ be the sub-path of $P$ from the starting point to $v$. Note that $v \notin V(P')$, which means $P'$ is also a path of the graph $G'$. Then we have that $\sum_{u\in V(P')} d_u' =  \sum_{u\in V(P')} d_u < \sum_{u\in V(P)} d_u <1$. This contradicts the fact that $\{d_u\}$ is an optimal solution of $\lpprimal{G',W}$. Therefore, $\{d_u\}$ is a valid solution of $\lpprimal{G,W}$ and thus $OPT_{lp}(G,W) \le \sum_{u \in V \setminus (W \cup \{v\})}d_u = \sum_{u \in V \setminus (W \cup \{v\})}d_u' = OPT_{lp}(G',W)$.
\end{proof}

\begin{lemma}[Complexity]
    Given an \abbrVMCcomp{} instance $I = (G,\mathbf{T},W,k)$, the algorithm $\verb|branching|$ will terminate in time $O^*(k^{O(k)})$ and output a set $\mathcal{I}$ of at most $k^{O(k)}$ instances.
    \label{lem:branching-complexity}
\end{lemma}

\begin{proof}
    Let $I = (G,\mathbf{T},W,k)$ be the input instance. Note that in each recursive call of $\verb|branching|$, the reduction and branching rules take only polynomial time and put at most one instance into the output $\mathcal{I}$.
    Therefore, to prove that the algorithm $\verb|branching|$ will terminate in time $O^*(k^{O(k)})$ and output at most $k^{O(k)}$ instances, it is sufficient to bound the number of leaves of the search tree.
    We then bound the number of leaves of the search tree by bounding the number of sub-instances in each branching rule and the depth of the search tree.
    
    We first bound the number of sub-instances in each branching rule. In Branching Rule 1, the number of sub-instances are 2. In Branching Rule 2, the number of sub-instances is bounded by $2\sum_{w \in W} |N(\FCfull{G}{W}{w}) \setminus W| \le 2k^2$.

    We use the measure $2k - OPT_{lp}(G,W)$ to bound the depth of the search tree. If $2k - OPT_{lp}(G,W) <0$, the algorithm will terminate by Reduction 1. Otherwise we have that $2k - OPT_{lp}(G,W) >0$. Suppose that there is a non-zero vertex $v \in V \setminus W$, we will either delete $v$ or make a torso of $v$ by Branching Rule 1. If $v$ is deleted, $k$ will decrease by 1 and it is easy to see that $OPT_{lp}(G,W)$ will decrease by at most 1. Thus the measure will decrease at least 1 in this case. If we make a torso of $v$, $k$ will not change and $OPT_{lp}(G,W)$ will increase by at least $\frac{1}{2}$ by Property 1 of Lemma~\ref{lem:OPT-torso-contra-vertex}. Thus the measure will decrease at least $\frac{1}{2}$ in this case. Therefore, in Branching Rule 1, the measure will decrease at least $\frac{1}{2}$.
    
    In Branching Rule 2, the algorithm will terminate if $\bigcup_{w \in W} N(\FCfull{G}{W}{w}) \setminus W = \emptyset$. Otherwise, we pick each $w \in W$ and $v \in N(\FCfull{G}{W}{w}) \setminus W$. We furtherly either delete $v$ or contract $v$ into $w$. Similarly, we know that the measure will decrease at least 1 if we delete $v$.   
    Let $I'=(G',\mathbf{T}',W,k)$ be the result of $\verb|contract|(I,f)$. We prove by contradiction that if the measure does not decrease when we contract $v$ into $w$ in Branching Rule 2, then there exists non-zero vertices in the graph $G'$. Assume that the measure does not decrease when we contract $v$ into $w$ in Branching Rule 2. Since Branching Rule 1 cannot be applied for the instance $I$, then $G$ does not contain any non-zero vertices. Thus, by Lemma~\ref{lem:2LP=mincut}, we have that $ 2OPT_{lp}(G,W) = \sum_{w_0\in W} \mincut(G,w_0,W\setminus\{w_0\})$. It is not hard to see $\mincut(G',w_0,W\setminus\{w_0\}) \ge \mincut(G,w_0,W\setminus\{w_0\})$ for any $w_0 \in W$. Since we contract $v \in N(\FCfull{G}{W}{w})$ into $w$ to obtain $I'$, it also holds that $\mincut(G',w,W\setminus\{w\})>\mincut(G,w,W\setminus\{w\})$.
    Therefore, we have that \[
        \sum_{w_0\in W} \mincut(G',w_0,W\setminus\{w_0\}) > \sum_{w_0\in W} \mincut(G,w_0,W\setminus\{w_0\})
    \]
    
    We claim that $G'$ has to contain a non-zero vertex.
    Suppose not, we have that $ 2OPT_{lp}(G',W) = \sum_{w_0\in W'} \mincut(G',w_0,W\setminus\{w_0\})$ by Lemma~\ref{lem:2LP=mincut}. Then we have that \[
        \begin{aligned}
            2OPT_{lp}(G',W) &= \sum_{w_0\in W} \mincut(G',w_0,W\setminus\{w_0\}) \\
            &>\sum_{w_0\in W} \mincut(G,w_0,W\setminus\{w_0\}) \\
            &=2OPT_{lp}(G,W).
        \end{aligned}
    \]
    Thus it holds that $2k - 2OPT_{lp}(G',W) < 2k - 2OPT_{lp}(G,W)$. This contradicts the fact that the measure does not decrease. Therefore, in this case, $G'$ must contain non-zero vertices. Then either the algorithm  will terminate or Branching Rule 1 will be applied. By Property 2 of Lemma~\ref{lem:OPT-torso-contra-vertex}, we have that $OPT_{lp}(G',W) \ge OPT_{lp}(G,W)$. Thus the measure of $I'$ is not greater than the measure of $I$. Combined with the half-integrality, we can say that if the algorithm does not terminate, the measure will decrease by at least $\frac{1}{2}$ in either this branching step or the next step. Thus the depth of the search tree can still be bounded by $O(k)$.
    
    Above all, the number of sub-instances of each branching rule can be bounded by $O(k^2)$ and the depth of the search tree can be bounded by $O(k)$. Therefore, the number of leaves of the search tree can be bounded by $k^{O(k)}$, which completes the proof.
\end{proof}

\begin{proof}[Proof of Lemma~\ref{lem:branch-into-contra}]
First we prove that the algorithm $\verb|branching|$ is a partial branching algorithm.
For each $I'=(G',\mathbf{T}',W',k')\in \mathcal{I}$, it is easy to check $|I'|\leq |I|^{O(1)}$ and $k'\leq k$.
Note that if $I$ is a \textbf{no}-instance, then the new instance $(G-\{v\},\mathbf{T}-\{v\},W,k-1)$ is a \textbf{no}-instance for any vertex $v \in V \setminus W$.
Similarly by Lemma~\ref{lem:torso} and Lemma~\ref{lem:contract}, the corresponding sub-instances are still \textbf{no}-instances if the original instance $I$ is a \textbf{no}-instance. Then by a conventional induction argument, one can easily prove that if $I$ is a \textbf{no}-instance, all instances in $\mathcal{I}$ are \textbf{no}-instances.

By Lemma~\ref{lem:branching-correctness}, we have that if $I$ is a \textbf{yes}-instance with a shadowless solution, then there is a \textbf{yes}-instance $I'=(G',\mathbf{T}',W',k')\in \mathcal{I}$ with a shadowless solution which is also contractible.

By Lemma~\ref{lem:branching-complexity}, we have that the algorithm $\verb|branching|$ will terminate in time $O^*(k^{O(k)})$ and output a set $\mathcal{I}$ of at most $k^{O(k)}$ instances.

Above all, we finish the proof of Lemma~\ref{lem:branch-into-contra}.
\end{proof}

\subsection{Proof of Lemma~\ref{lem:2LP=mincut}}\label{subsec:MeasureLemma}
The goal of this subsection is to prove Lemma~\ref{lem:2LP=mincut}, restated below for convenience. 
\lemLPeqMincut*
We begin with the following direction.
\begin{lemma}\label{lem:2LP<=mincut}
    It holds that 
    \[
    2OPT_{lp}(G,W)\le \sum_{w\in W} \mincut(G,w,W\setminus\{w\}).
    \]
\end{lemma}
\begin{proof}
For each $w\in W$, let $N(C_w)$ be a minimum separator (where $C_w$ is the reachable set) that separates $w$ from $W \setminus \{w\}$, so $|N(C_w)| = \mincut(G,w,W\setminus\{w\})$.
Consider a solution $d: V \setminus W \to \mathbb{R}_{\ge 0}$ to the primal LP $\lpprimal{G,W}$ defined by:
\[
    d_v = \begin{cases} 
        1 & \text{if } v\in N(C_w) \text{ for at least two $w\in W$,}\\
        1/2 & \text{if } v \in N(C_w) \text{ for exactly one $w\in W$}, \\
        0 & \text{otherwise.}
    \end{cases}
\]
To see that $\{d_v\}$ is feasible, consider any path $P$ connecting distinct $w_i,w_j\in W$. The path must traverse at least one vertex in $N(C_{w_i})$ and at least one vertex in $N(C_{w_j})$. If $P$ intersects $N(C_{w_i})$ and $N(C_{w_j})$ at a common vertex $v$, then $d_v=1$ and so $\sum_{z\in P}d_{z}\geq d_v=1$. Otherwise, $P$ contains distinct vertices $u\in N(C_{w_i})$ and $v\in N(C_{w_j})$, implying $\sum_{z\in P}d_{z}\geq d_u+d_v\geq 1/2+1/2\geq 1$.
Since $\{d_v\}$ is feasible, we have 
\[ 
OPT_{lp}(G,W)\le \sum_{v\in V\setminus W} d_v\leq \frac{1}{2}\sum_{v\in W}|N(C_w)|=\frac{1}{2}\sum_{w\in W} \mincut(G,w,W\setminus\{w\}),
\] 
which completes the proof.
\end{proof}

It remains to show $2OPT_{lp}(G,W)\ge \sum_{w\in W} \mincut(G,w,W\setminus\{w\})$.
To this end, we exploit the structural properties of the dual LP $\lpdual{G,W}$. In particular, our argument relies on a canonical choice of an optimal dual solution that maximizes the set of vertices with non-binding capacity constraints.
For convenience, a vertex $v \in V \setminus W$ is called \emph{untight} with respect to a solution $\{f_P\}$ to the dual LP $\lpdual{G,W}$ if its capacity constraint is not binding, i.e., $\sum_{P: v \in P} f_P < 1$. 
Otherwise, $v$ is called \emph{tight}.

Let $\{f^*_P\}$ be an optimal solution to the dual LP that maximizes the number of untight vertices, and let $S^*$ denote this set of untight vertices.
For each $w \in W$, we define the \emph{untight region} $U_w\subseteq S^*$ as the set of vertices in $V \setminus W$ that are reachable from $w$ via a path consisting entirely of untight vertices. Let $N(U_w)$ denote the open neighborhood of $U_w$, i.e., $N(U_w) := \{v \notin U_w \mid \exists u \in U_w, (u,v) \in E(G) \}$. By this definition, all vertices in $N(U_{w})$ are tight.
We have the following property for untight regions.
\begin{lemma}\label{lem:N(U)-disjoint}
    If there is no non-zero vertex, then $N(U_{w_i}) \cap N(U_{w_j}) = \emptyset$ holds for any two distinct $w_i, w_j \in W$.
\end{lemma}
\begin{proof}
Suppose for the sake of contradiction that there exists a vertex $v \in N(U_{w_i}) \cap N(U_{w_j})$.
Then $v$ is a tight vertex.
Since $v \in N(U_{w_i}) \cap N(U_{w_j})$, there exist neighbors $u_i \in U_{w_i}$ and $u_j \in U_{w_j}$ of $v$. By the definition of untight regions, there is a path from $w_i$ to $u_i$ consisting entirely of (untight) vertices in $U_{w_i}$, and similarly for $w_j$ and $u_j$. Concatenating these segments through $v$, we obtain a path $P$ connecting $w_i$ and $w_j$ where $v$ is the only tight vertex on the path. Consider any optimal solution $\{d^*_v\}$ to the primal LP $\lpprimal{G,W}$. By complementary slackness, $d^*_u = 0$ for every untight vertex $u$; in particular, $d^*_u=0$ for all vertices $u\in P\setminus \{v\}$. Together with the feasibility constraint $\sum_{z \in P \setminus W} d^*_z \ge 1$ for path $P$, we have that $d^*_v\geq 1$ in \emph{every} optimal solution to $\lpprimal{G,W}$.
However, this implies that $v$ is a non-zero vertex, a contradiction. 
\end{proof}

\begin{lemma}\label{lem:path-pass-two-N(U)}
    Assume there is no non-zero vertex. Let $P$ be any simple path connecting distinct $w_i, w_j \in W$ with $f^*_{P} > 0$. Then $P$ intersects the set $\bigcup_{w \in W} N(U_w)$ at exactly two distinct vertices, one in $N(U_{w_i})$ and one in $N(U_{w_j})$.
\end{lemma}
\begin{proof}
We prove by contradiction. Suppose there exists a simple path $P$ connecting $w_i, w_j$ with $f^*_{P} > 0$ that violates the statement.
Since $P$ originates in $U_{w_i}$ and terminates in $U_{w_j}$, it must traverse at least one vertex in $N(U_{w_i})$ to exit $U_{w_i}$ and at least one vertex in $N(U_{w_j})$ to enter $U_{w_j}$.
By Lemma~\ref{lem:N(U)-disjoint}, $N(U_{w_i})\cup N(U_{w_j})=\emptyset$, implying $|P \cap (N(U_{w_i}) \cup N(U_{w_j}))| \ge 2$.
Consequently, a violation can only occur in one of the following two cases:
\begin{itemize}
    \item Case~1: The path $P$ intersects the boundaries of exactly two untight regions (i.e., $N(U_{w_i})$ and $N(U_{w_j})$), but $|P \cap (N(U_{w_i}) \cup N(U_{w_j}))| \ge 3$.
    \item Case~2: The path $P$ intersects $N(U_{w_k})$ for some $w_k\in W\setminus \{w_i,w_j\}$.
\end{itemize}

For Case~1, without loss of generality, assume $P$ intersects $N(U_{w_i})$ at two distinct vertices $u_1$ and $u_2$, where $u_2$ is the \emph{last} vertex of $N(U_{w_i})$ visited by $P$ and $u_1$  appears before $u_2$, along the direction from $w_i$ to $w_j$.
By the definition of the untight region $U_{w_i}$, there exists a path $Q$ from $w_i$ to $u_2$ such that all internal vertices of $Q$ belong to $U_{w_i}$ (and are thus untight).

We construct a new path $P'$ by concatenating $Q$ with the suffix of $P$ (from $u_2$ to $w_j$, denoted by $P_{u_2 \to w_j}$).
Crucially, $P'$ does not contain $u_1$. 
We modify $\{f^*_P\}$ to $\{f'_P\}$ as follows:
\[
    f'_{P} = f^*_{P}-\epsilon, f'_{P'}=f^*_{P'}+\epsilon, \text{ and } f'_{R}=f^*_{R} \text{ for all other $R\in \mathcal{P}(G,W)$},
\]
where $\epsilon$ is a sufficiently small positive value defined by
\[
\epsilon=\frac{1}{2} \min \left( f^*_P, \min_{v \in V(Q) \setminus \{u_2\}} \left( 1 - \sum_{R: v \in R} f^*_R \right) \right).\]
Since all internal vertices of $Q$ are untight, the slack term $(1-\sum_{R:v\in R}f^*_{R})$ is positive. Together with $f^*_P>0$, we have $\epsilon > 0$. We claim that the new $\{f'_P\}$ is feasible. The reason is as follows. 
For vertices on the suffix $P_{u_2 \to w_j}$, the flow decrease of $P$ and flow increase of $P'$ cancel out.
For internal vertices of $Q$, the capacity constraint is satisfied by the choice of $\epsilon$. For $u_2$, the total flow remains unchanged. Thus, $\{f'_P\}$ is feasible (and optimal since $\sum_{P\in \mathcal{P}(G,W)} f^*_P = \sum_{P\in \mathcal{P}(G,W)} f'_P$).

However, consider the tight vertex $u_1$. The original path $P$ passes through $u_1$, but the new path $P'$ bypasses it. Thus, the total load on $u_1$ strictly decreases, i.e.,
\[
    \sum_{R: u_1 \in R} f'_R \leq \sum_{R: u_1 \in R} f^*_R - \epsilon = 1 - \epsilon < 1.
\]
This implies that $u_1$ is an untight vertex with respect to $\{f'_{P}\}$.
Since we only increased flow on the internal vertices in $Q$, and these vertices remain untight due to the choice of $\epsilon$, the set of untight vertices for $\{f'_P\}$ is a strict superset of that for $\{f^*_P\}$. This contradicts the maximality of the untight set in the optimal solution $f^*$.

For Case~2, $P$ intersects not only $N(w_i)$ and $N(w_j)$, but also $N(w_k)$ for some  $w_k\in W\setminus\{w_i,w_j\}$. Let $u \in P \cap N(U_{w_k})$ be a vertex in this intersection.
By the definition of the untight region $U_{w_k}$, there exists a path $Q$ connecting $w_k$ to $u$ such that all internal vertices belong to $U_{w_k}$. 

Let $\{d^*_v\}$ be any optimal solution to the primal LP $\lpprimal{G,W}$. By complementary slackness, $d^*_v = 0$ for all $v \in Q \setminus \{u\}$. We decompose the path $P$ at vertex $u$ into two segments $P_{w_i \to u}$ (from $w_i$ to $u$) and $P_{w_j \to u}$ (from $w_j$ to $u$) and construct two new paths:
\begin{itemize}
    \item path $R_{w_i\to w_k}$ connecting $w_i$ and $w_k$ by concatenating $P_{w_i\to u}$ with $Q$ (from $u$ to $w_k$), and
    \item path $R_{w_j\to w_k}$ connecting $w_j$ and $w_k$ by concatenating $P_{w_j\to u}$ with $Q$ (from $u$ to $w_k$).
\end{itemize}
Since $d^*$ is a feasible primal solution, together with $d^*_v = 0$ for all $v \in Q \setminus \{u\}$, we have 
\begin{align}\label{eq:two-seg-lb}
    \sum_{v \in P_{w_i\to u}} d^*_v = \sum_{v \in R_{w_i\to w_k}} d^*_v \ge 1 \quad \text{and} \quad \sum_{v \in P_{w_j\to u}} d^*_v=\sum_{v \in R_{w_j\to w_k}} d^*_v \ge 1.
\end{align}
Since $f^*_P > 0$, by complementary slackness and substituting Eq.~\eqref{eq:two-seg-lb}, we have
\[
    1 = \left( \sum_{v \in P_{w_i \to u}} d^*_v \right) + \left( \sum_{v \in P_{w_j\to u}} d^*_v \right) - d^*_u \ge 1 + 1 - d^*_u = 2 - d^*_u,
\]
which simplifies to $d^*_u \ge 1$. That is, $d^*_u\geq 1$ in \emph{every} optimal solution to $\lpprimal{G,W}$.
However, this implies that $u$ is a non-zero vertex, a contradiction. 
\end{proof}

\begin{lemma}\label{lem:2LP>=mincut}
    If there is no non-zero vertex, then
    \[
    2OPT_{lp}(G,W)\ge \sum_{w\in W} \mincut(G,w,W\setminus\{w\}).
    \]
\end{lemma}
\begin{proof}

    Consider the optimal solution $\{f^*_{P}\}$ to the dual LP $\lpdual{G,W}$. Recall that for every $v \in \bigcup_{w\in W}N(U_w)$, the capacity constraint is binding, i.e., $\sum_{P: v \in P} f^*_P = 1$.
    Thus, we have
    \[
        \sum_{w \in W} |N(U_w)| = \sum_{w \in W} \sum_{v \in N(U_w)} 1 = \sum_{w \in W} \sum_{v \in N(U_w)} \left( \sum_{P: v \in P} f^*_P \right).
    \]
    We then exchange the order of summation to iterate over flow paths. Let $P$ be a simple path connecting $w_i$ and $w_j$ with $f^*_P > 0$. By Lemma~\ref{lem:path-pass-two-N(U)}, such a path intersects $\bigcup_{w\in W} N(U_w)$ exactly twice: once at a vertex in $N(U_{w_i})$ and once at a (different) vertex in $N(U_{w_j})$. Thus, each flow variable $f^*_P$ is counted exactly twice in the summation, and we have
    \begin{align}\label{eq:sum_neighbor=2LP}
        \sum_{w \in W} |N(U_w)|= \sum_{w \in W} \sum_{v \in N(U_w)} \left( \sum_{P: v \in P} f^*_P \right) = \sum_{\substack{P \in \mathcal{P}(G,W)\\ f^{*}_{P}>0}} 2 f^*_P = 2 OPT_{lp}(G,W).
    \end{align}
    
    On the other hand, for each $w \in W$, $N(U_w)$ is a separator that separates $w$ from $W \setminus \{w\}$ (since any path from $w$ to another vertex in $W \setminus \{w\}$ must exit $U_w$). Thus, we have $|N(U_w)| \ge \mincut(G, w, W \setminus \{w\})$. Substituting this into Eq.~\eqref{eq:sum_neighbor=2LP} yields
    \[
        2OPT_{lp}(G,W) = \sum_{w\in W}|N(U_{w})|\ge \sum_{w \in W} \mincut(G, w, W \setminus \{w\}),
    \]
    which completes the proof.
\end{proof}

Lemma~\ref{lem:2LP=mincut} directly follows from Lemma~\ref{lem:2LP<=mincut} and Lemma~\ref{lem:2LP>=mincut}.

\subsection{Contracting Contractible Instances}\label{subsec:contract}
We now describe how to convert an instance $I$ of \abbrVMCcomp{} into a bipedal instance $I'$.
Note that we do not require that $I$ has a solution if and only if $I'$ does. Instead, we will pay more attention to the case that $I$ admits a contractible solution.

\begin{restatable}{lemma}{lembipedal}\label{lem:compti-reduce-to-bipedal}
    There exists a polynomial-time algorithm that takes an \VMCcomp{} instance $I=(G, \mathbf{T}, W, k)$, and outputs a bipedal instance $I' = (G',\mathbf{T}',W',k')$, such that:
    \begin{enumerate}[(1)]
        \item {if $I$ admits a shadowless and contractible solution, then $I'$ admits a shadowless solution;}\label{condition:1:lem:compti-reduce-to-bipedal}
        \item if $I'$ is a \textbf{yes}-instance, then $I$ is a \textbf{yes}-instance;\label{condition:2:lem:compti-reduce-to-bipedal}
        \item $|V(G')|\leq |V(G)|$ and $k'\leq k$;\label{condition:3:lem:compti-reduce-to-bipedal}
    \end{enumerate}

\end{restatable}

\subsubsection{Regions}

\begin{definition}[Regions]
    Let $G=(V, E)$ be a graph and $W\subseteq V$ be a vertex set. 
    For a vertex set $R\subseteq W$, we define 
    \[
    \Partfull{W}{G}{R}:=\bigcap_{w\in R}\FCfull{G}{W}{w} \setminus \bigcup_{w'\in W\setminus R}\FCfull{G}{W}{w'}.
    \]
    We may simply write it as $\Part{R}$ when $G$ is clear in the context.
\end{definition}

Note that $\Partfull{W}{G}{\emptyset}$ may be nonempty.  
The sets $\Partfull{W}{G}{R}$ for all subsets $R \subseteq W$ (including $\emptyset$) form a partition of the vertex set $V$.
An illustration of regions for $|W|=3$ is shown in Figure~\ref{fig:Area}.

\begin{figure}[!h]
    \centering
    \includegraphics[width=0.75\textwidth]{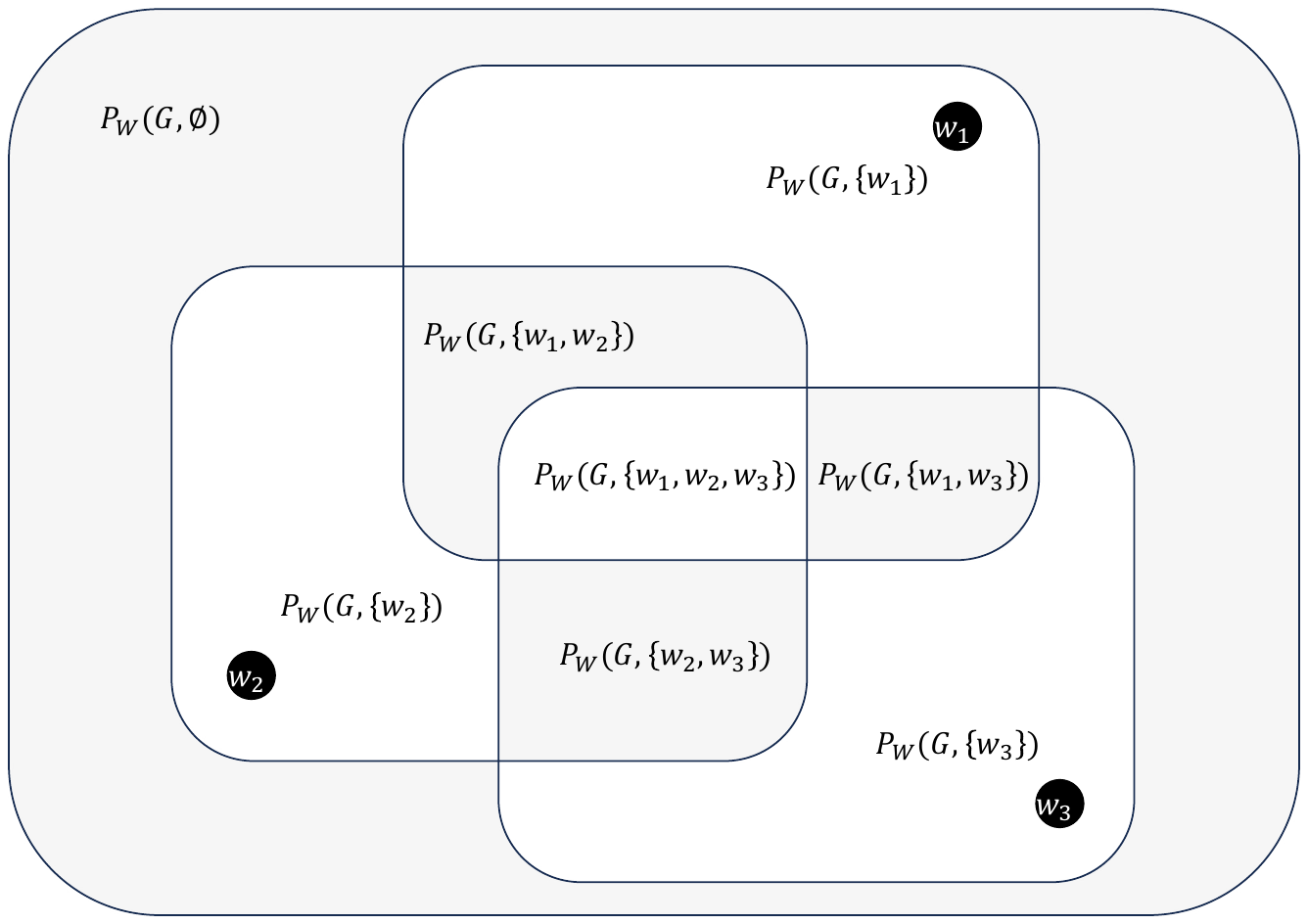}
    \caption{An example for regions, where $W=\{w_1,w_2,w_3\}$. The diagram shows how the $2^3=8$ regions $\Partfull{W}{G}{R}$ (one for each $R\subseteq W$) are formed by the intersections of $\FCfull{G}{W}{w_1}$, $\FCfull{G}{W}{w_2}$, and $\FCfull{G}{W}{w_3}$ (and their complements). Together, these $8$ disjoint regions form a complete partition of the entire vertex set.}
    \label{fig:Area}
\end{figure}

We first prove the following property.

\begin{lemma}\label{lem:compactible-p-empty}
     Let $I=(G=(V,E),\mathbf{T},W,k)$ be an instance of the \abbrVMCcomp{} problem.
     If $I$ admits a shadowless and contractible solution, then $\Partfull{W}{G}{\emptyset}=\emptyset$.
\end{lemma}

\begin{proof}

     Suppose $\Partfull{W}{G}{\emptyset} \neq \emptyset$, then let $v$ be a vertex in $\Partfull{W}{G}{\emptyset}$. Let $X$ be an optimal shadowless and contractible solution of $I$. Then either $v\in X$ or not. If $v\notin X$, then since $X$ is shadowless, there is a path connecting $v$ to some $w\in W$ in $G-X$. Since $v\in \Partfull{W}{G}{\emptyset}\subseteq  V(G) \setminus \FC{w}$, such a path contains a subpath connecting a vertex in $N(\FC{w})$ and $w$, contradicting that $X$ is contractible. If $v\in X$, we claim that $X\setminus \{v\}$ is still a shadowless and contractible solution, contradicting that $X$ is optimal. Indeed, any path connecting $v$ and some $w\in W$ has to contain a subpath, which is an $\FC{w}$-path connecting some $v'\in N(\FC{w})$ and $w$. Since $X$ is contractible, this actually implies that any path connecting $v$ and $W$ is hit by a vertex in $X\setminus \{v\}$. Any path connecting a pair in $\mathbf{T}$ or two vertices in $W$ is hit by some vertex in $W$, implying that $X\setminus\{v\}$ is a solution. Similarly, we can easily verify that $X\setminus\{v\}$ is contractible and shadowless\footnote{Here if $v$ is in a component $CC$ in $G$, and $CC\subseteq X$, our argument is not complete. However, it is actually fine to assume that $X$ is optimal, because our shadow removal procedure and branching algorithm preserves optimality. Another way is to add a reduction rule to avoid this trivial case (in fact never appears).}.
\end{proof}

The following property of regions is crucial in our proof.

\begin{lemma}[Property of Regions]\label{lem:prop-area}
    Let $G=(V,E)$ be a graph and $W$ be a vertex set such that $\Partfull{W}{G}{\emptyset}=\emptyset$.
    For any two distinct subsets $R_1, R_2\subseteq W$, 
    if there exists an edge between a vertex in $\Partfull{W}{G}{R_1}$ and a vertex in $\Partfull{W}{G}{R_2}$, 
    then one of the following conditions holds:
    \begin{enumerate}[(1)]
        \item {$|R_1| = |R_2| = 1$}; or
        \item $|R_1| = 1$, $|R_2| = 2$, and $R_1 \subseteq R_2$; or
        \item $|R_1| = 2$, $|R_2| = 1$, and $R_2 \subseteq R_1$.
    \end{enumerate}
\end{lemma}

    

\begin{proof}

    Our proof relies on the following two properties.

    \textbf{Property (a).}
    For any $w_i,w_j\in W$,
    \[
        N(\FC{w_i})\cap N(\FC{w_j}) = \emptyset,
    \]

    \textbf{Property (b).}
    \[
        \bigcup_{w\in W} \Bnd{\Part{\{w\}}} = \bigcup_{w\in W} N(\FC{w}).
    \]

    We first prove these two properties by establishing the following chain of inequalities.

    Since $\Part{\emptyset} = \emptyset$, we know that $\bigcup_{w\in W}\Bnd{\Part{\{w\}}}\subseteq \bigcup_{w\in W} N(\FC{w})$. 
    Thus, it holds that
    \begin{equation}
        \left|\bigcup_{w\in W}\Bnd{\Part{\{w\}}}\right| \leq \left|\bigcup_{w\in W} N(\FC{w})\right|\leq\sum_{w\in W} |N(\FC{w})|. 
        \label{eq:1}
    \end{equation}


    For each $w\in W$, by the definition of $\FC{w}$, 
    $N(\FC{w})$ is a minimum separator that separates $\{w\}$ from $(W\setminus \{w\})$.
    Since $\Bnd{\Part{\{w\}}}$ is a separator that separates $\{w\}$ from $(W\setminus \{w\})$, we have $|N(\FC{w})|\leq |\Bnd{\Part{\{w\}}}|$.
    Thus, it holds that 
    \begin{equation}
        \sum_{w\in W} |N(\FC{w})|\leq \sum_{w\in W} |\Bnd{\Part{\{w\}}}|.
        \label{eq:3}
    \end{equation}

    By the definition of $\Part{\{w\}}$ for every $w\in W$, we know that for each $w_i, w_j\in W$, it holds that $\Part{\{w_i\}}\cap \Part{\{w_j\}} = \emptyset$.
    Thus, it holds that
    \begin{equation}
        \sum_{w\in W} |\Bnd{\Part{\{w\}}}| = \left|\bigcup_{w\in W}\Bnd{\Part{\{w\}}}\right|.
        \label{eq:4}
    \end{equation}

    Combining the Eqs.~\eqref{eq:1}\eqref{eq:3}\eqref{eq:4}, we have that
    \begin{equation}
        \left|\bigcup_{w\in W}\Bnd{\Part{\{w\}}}\right| \leq 
        \sum_{w\in W} |N(\FC{w})| \leq
        \sum_{w\in W} |\Bnd{\Part{\{w\}}}| = 
        \left|\bigcup_{w\in W}\Bnd{\Part{\{w\}}}\right|.
        \label{eq:5}
    \end{equation}
    Thus, all inequalities hold with equality simultaneously.
    In particular, 
    \[
    \left|\bigcup_{w\in W} N(\FC{w})\right| = 
        \sum_{w\in W} |N(\FC{w})|,
    \]
    which implies that for any $w_i,w_j\in W$,
    \[
        N(\FC{w_i})\cap N(\FC{w_j}) = \emptyset.
    \]
    Thus, Property (a) holds.

    By Eq.~\eqref{eq:5}, we also know that $\left|\bigcup_{w\in W}\Bnd{\Part{\{w\}}}\right| = \left|\bigcup_{w\in W} N(\FC{w})\right|$.
    Recall that since $\Part{\emptyset} = \emptyset$, we know that $\bigcup_{w\in W}\Bnd{\Part{\{w\}}}\subseteq \bigcup_{w\in W} N(\FC{w})$.
    Thus, it holds that
    \[
        \bigcup_{w\in W} \Bnd{\Part{\{w\}}} = \bigcup_{w\in W} N(\FC{w}).
    \]
    Thus, Property (b) holds.

    Now, we are ready to prove the theorem.
    It is sufficient to show that for any $R, R'\subseteq W$ such that (i) $|R\setminus R'|\geq 2$ or (ii) $|R\setminus R'|\leq 1$ and $|R| + |R'| \geq 3$, there is no edge between $R$ and $R'$. 
    By contradiction, we assume that there exists an edge $(u, v)$ such that $u\in \Part{R}$ and $v\in \Part{R'}$.
    We consider the following two cases.
    
    Case (i). $|R\setminus R'|\geq 2$. 
    Let $w_1$ and $w_2$ be two vertices in $R\setminus R'$. 
    By the definition of $\Part{\cdot}$, we know that $u\in N(\FC{w_1})$ and $u\in N(\FC{w_2})$.
    However, by Property (a), it holds that $N(\FC{w_1})\cap N(\FC{w_2}) = \emptyset$, which leads to a contradiction.

    Case (ii). $|R\setminus R'|\leq 1$ and $|R| + |R'| \geq 3$.
    In this case, it is not hard to see that $|R|\geq 2$ and $|R'|\geq 2$.
    Let $w$ be a vertex in $R\setminus R'$.
    By the definition of $\Part{\cdot}$, we know that $v\in N(\FC{w})$.
    However, since $v\in \Part{R'}$ where $|R'|\geq 2$, it holds that $v\notin \bigcup_{w\in W} \Bnd{\Part{\{w\}}}$.
    Thus, it holds that 
    $\bigcup_{w\in W} \Bnd{\Part{\{w\}}} \neq \bigcup_{w\in W} N(\FC{w})$, which contradicts Property (b).

    Thus, this theorem holds.
\end{proof}

\begin{corollary}\label{coro-area-3}
    Let $G=(V,E)$ be a graph and $W$ be a vertex set such that $\Partfull{W}{G}{\emptyset}=\emptyset$.
    For any $R\subseteq W$ with $|R| \geq 3$, we have that $\Partfull{W}{G}{R} = \emptyset$.
\end{corollary}

\begin{proof}
    
    Suppose for the sake of contradiction that there exists $R\subseteq W$ with $|R|\geq 3$ such that $\Partfull{W}{G}{R} \neq \emptyset$.
    By the definition of the regions, $\Partfull{W}{G}{R} \subseteq \bigcap_{w\in R}\FCfull{G}{W}{w}$. Thus, each vertex in $\Partfull{W}{G}{R}$ must be connected to some vertex $w\in R$.
    That is, there is a path from a vertex in $\Partfull{W}{G}{R}$ to vertex $w$.
    Such a path should contain an edge between $\Partfull{W}{G}{R}$ and $\Partfull{W}{G}{R'}$ for some $R'\supseteq \{w\}$ s.t. $R'\neq R$, since $w\in \Partfull{W}{G}{\{w\}} $ and $\Partfull{W}{G}{R}\cap \Partfull{W}{G}{\{w\}}=\emptyset$.
    However, since $|R|\geq 3$, by Lemma~\ref{lem:prop-area}, there should not exist an edge between $\Partfull{W}{G}{R}$ and $\Partfull{W}{G}{R'}$ for any $R'\subseteq W$.  
    A contradiction.
    Therefore, $\Partfull{W}{G}{R} = \emptyset$ for all $|R|\geq 3$.
\end{proof}

\subsubsection{The algorithm}

\begin{step}\label{step:pempty}
    If $\Partfull{W}{G}{\emptyset}\neq\emptyset$, then return a trivial bipedal no-instance (e.g., $(\emptyset,\emptyset,\emptyset,-1)$).
\end{step}

\begin{step}\label{step:contract}
    Let $f:\bigcup_{w\in W}\Bnd{\Partfull{W}{G}{\{w\}}}\setminus W\to W$ be a function where $f(v)=w$ if $v\in \Bnd{\Partfull{W}{G}{\{w\}}}$
    (note that $f$ is well-defined, since for any distinct $w',w\in W$, $\Bnd{P_{W}(G,\{w\})}\cap \Bnd{P_{W}(G,\{w'\})}=\emptyset$).
    Let $(G_{f}, \textbf{T}_f, W, k)$ be the resulting instance of $\verb|contract|(I,f)$.
    Return $(G_{f}, \textbf{T}_f, W, k)$.
\end{step}

\begin{lemma}\label{lem:bipedal}
    Let $I=(G=(V,E),\mathbf{T},W,k)$ be an instance of the \VMCcomp{} problem, where $\Partfull{W}{G}{\emptyset}=\emptyset$.
    It holds that $(G_{f}, \mathbf{T}_{f}, W, k)$ is a bipedal instance.
\end{lemma}

\begin{proof}

    By Corollary~\ref{coro-area-3}, it holds that $\Part{R}=\emptyset$ for any $|R| \geq 3$. 
    Since $\Part{\emptyset} = \emptyset$, the vertices in $G$ must be in $\Part{R}$ where $1\leq |R| \leq 2$.
    
    \begin{figure}[!t]
        \centering
        \includegraphics[width=0.7\textwidth]{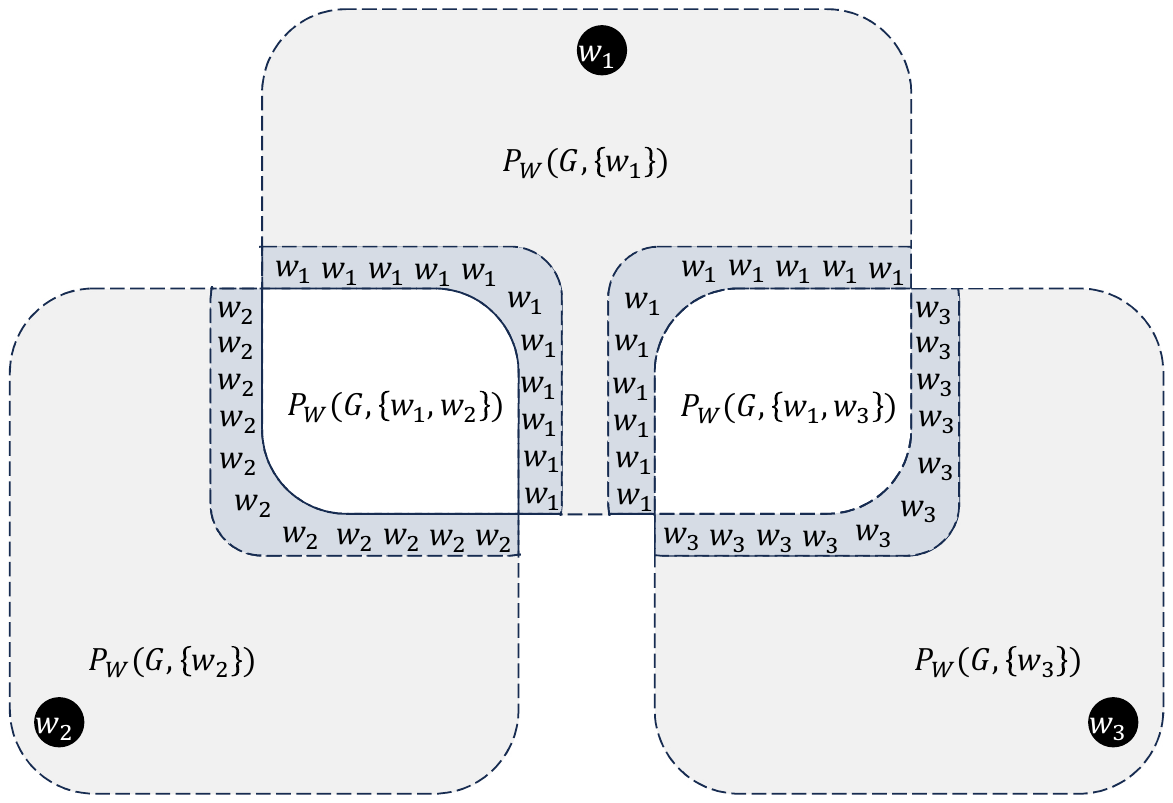}
        \caption{An illustration for the contraction in the proof of Lemma \ref{lem:bipedal}, where $W = \{w_1,w_2,w_3\}$.
        In this contraction, the vertices in $\Bnd{\Part{\{w_i\}}}$ are contracted into $w_i$ for each $w_i\in W$.
        }
        \label{fig:Contraction}
    \end{figure}

    Recall that $f:\bigcup_{w\in W}\Bnd{\Part{\{w\}}}\setminus W\to W$ is a function where $f(v)=w$ if $v\in \Bnd{\Part{\{w\}}}$.
    And we contract $f^{-1}(w)$ into $w$ for each $w\in W$.
    See Figure~\ref{fig:Contraction} for an illustration.
    By the definition of $f(\cdot)$,
    the connected components of $G_f - W$ are exactly the connected components of $G - (W \cup  \bigcup_{w\in W}\Bnd{\Part{\{w\}}})$.
    Consider a vertex $v$ in $G_f-W$. Let $\CC{v}$ denote the component containing $v$ in $G_f- W$.
    If $v\in \Part{\{w\}}$ for some $w$, then since $\Bnd{\Part{\{w\}}}$ is contracted into $w$, $N_{G_f}(CC_v)\cap W\subseteq \{w\}$, thus $|N_{G_f}(CC_v)\cap W|\leq 1$. Otherwise, by Corollary~\ref{coro-area-3} and $\Part{\emptyset} = \emptyset$, $v\in \Part{\{w_1,w_2\}}$ for some $w_1,w_2\in W$. 
    By Lemma~\ref{lem:prop-area}, $N(\Part{\{w_1,w_2\}})\subseteq \Part{\{w_1\}}\cup \Part{\{w_2\}}$. Since $\Bnd{\Part{\{w_1\}}}$ is contracted into $w_1$ and $\Bnd{\Part{\{w_2\}}}$ is contracted into $w_2$, $N_{G_f}(CC_v)\cap W\subseteq \{w_1,w_2\}$, thus $|N_{G_f}(CC_v)\cap W|\leq 2$. This completes the proof.
\end{proof}

Next, we show the correspondence of the solutions of the input \abbrVMCcomp{} instance and the output bipedal instance.

\begin{lemma}\label{lem:compactible-property}
    Let $I=(G=(V,E),\mathbf{T},W,k)$ be an instance of the \VMCcomp{} problem.
    It holds that
    \begin{enumerate}[(1)]
        \item {if $I$ admits a shadowless and contractible solution, then $(G_{f}, \mathbf{T}_{f}, W, k)$ admits a shadowless solution};
        \item if $(G_{f}, \mathbf{T}_{f}, W, k)$ is a \textbf{yes}-instance, then $I$ is a \textbf{yes}-instance.
    \end{enumerate}
\end{lemma}

\begin{proof}


    We first prove statement (1). 
    Assume that $I$ admits a shadowless and contractible solution and let $X$ be the shadowless and contractible solution of $I$.
    Let $\CC{v}$ for a vertex $v\in V$ denote the component containing $v$ in $G- X$.
    Let $w$ be an arbitrary vertex in $W$ and $v$ be an arbitrary vertex in $\Bnd{\Part{\{w\}}}$. 
    We have the following two claims.
    
    Claim~1: $\CC{v}\cap (W\setminus \{w\})=\emptyset$. 
    Since $X$ is a contractible solution, we have $\CC{w}\cap N(\FC{w})= \emptyset$.
    Suppose for the sake of contradiction that there is some $w'\in W\setminus\{w\}$ such that $v\in \CC{w'}$. Then there must be a path from $v$ to $w'$ in $G- X$. Since $v\notin \FC{w'}$,  such a path should contain at least one vertex in $N(\FC{w'})$, and this implies that $\CC{w'}\cap N(\FC{w'})\neq \emptyset$. A contradiction. 

    Claim~2: $\CC{v}\cap W\neq \emptyset$.
    It clearly holds since $X$ is a shadowless solution.
    
    Combining the above two claims, we have that for any $w\in W$ and $v\in \Bnd{\Part{\{w\}}}$, $\CC{v}\cap W=\{w\}$.
    That is, for any $w\in W$, all vertices in 
     $\Bnd{\Part{\{w\}}}$ and $w$ are in the same component in $G - X$. 
     This means that $X$ is also a solution of $(G_{f}, \mathbf{T}_{f}, W, k)$.
     By Lemma~\ref{lem:shadow-correctness}, $X$ is a shadowless solution of $(G_{f}, \mathbf{T}_{f}, W, k)$. 

     For statement (2), it is clear that any solution to $(G_{f}, \mathbf{T}_{f}, W, k)$ can be transformed into a solution to $I$ by reversing the corresponding contraction operations that yield $G_{f}$ and $\mathbf{T}_{f}$.
\end{proof}

The proof of Lemma \ref{lem:compti-reduce-to-bipedal} is given below.

\begin{proof}
 [Proof of Lemma~\ref{lem:compti-reduce-to-bipedal}]

    For an input \abbrVMCcomp{} instance $I$,
    our algorithm always outputs a bipedal instance:
    in Step~\ref{step:pempty}, it outputs a trivial bipedal no-instance; 
    in Step~\ref{step:contract}, 
    since Step~\ref{step:pempty} cannot be applied, it holds that $\Part{\emptyset} = \emptyset$. By Lemma~\ref{lem:bipedal}, the instance output in Step~\ref{step:contract} is a bipedal instance.
    Next, we check the three conditions.

    For condition~(\ref{condition:1:lem:compti-reduce-to-bipedal}), if $I$ admits a shadowless and contractible solution, by Lemma~\ref{lem:compactible-property}, we know that the output instance $(G_f,\textbf{T}_f, W, k)$ admits a shadowless solution. Thus, condition~(\ref{condition:1:lem:compti-reduce-to-bipedal}) holds.
    For condition~(\ref{condition:2:lem:compti-reduce-to-bipedal}), by Lemma~\ref{lem:compactible-property}, we know that $I$ is a \textbf{yes}-instance, and thus condition~(\ref{condition:2:lem:compti-reduce-to-bipedal}) holds.
    For condition~(\ref{condition:3:lem:compti-reduce-to-bipedal}), it can be verified that the number of vertices and the parameter $k$ do not increase.

    Finally, we analyze the running time.
    For Step~\ref{step:contract}, contracting vertices to obtain $G_{f}$ and $\mathbf{T}_{f}$ can be done in $n^{O(1)}$ time. 
    Thus, our algorithm runs in $n^{O(1)}$ time.
\end{proof}

\section{Conclusion and Discussion}\label{sec:conclusion}
In this work, we presented a $k^{O(k)} n^{O(1)}$-time algorithm for the \textsc{Vertex Multicut} problem, improving upon the previous best $k^{O(k^2)} n^{O(1)}$-time result by Chitnis et al.~\cite{DBLP:journals/talg/ChitnisCHM15-FPT-DSFVS}. Our main technical contribution is a refined shadow removal procedure. 
Together with a deep analysis of the reduction to bipedal instances, this yields a faster algorithm while preserving the overall framework of Marx and Razgon~\cite{DBLP:journals/siamcomp/MarxR14-first-FPT}.  
Our refined shadow removal procedure also implies a $k^{O(k^2)}n^{O(1)}$ time algorithm for {\sc Directed Subset Feedback Vertex Set} and a $k^{O(k)}n^{O(1)}$ time algorithm for {\sc Directed Multiway Cut}, improving over the previously best known algorithms of Chitnis et al. ~\cite{DBLP:journals/talg/ChitnisCHM15-FPT-DSFVS}.
We believe that the refined shadow removal technique and other techniques developed in this paper can be applied to other parameterized cut and separation problems, particularly those that can be expressed in terms of bipedal compression or shadowless solutions.

An interesting direction for future work is to explore the possibility of obtaining single-exponential parameterized algorithms for \textsc{Edge Multicut}, 
that is, algorithms running in time $2^{O(k)} n^{O(1)}$. Another problem is to improve over the bound in the shadow removal theorem, or to prove a lower bound.

\bibliography{bib_file}

@inproceedings{2-multicut-YannakaMKPCSPC1983,
author = {Yannakakis, Mihalis and Kanellakis, Paris C. and Cosmadakis, Stavros S. and Papadimitriou, Christos H.},
title = {Cutting and Partitioning a Graph aifter a Fixed Pattern (Extended Abstract)},
year = {1983},
isbn = {3540123172},
publisher = {Springer-Verlag},
address = {Berlin, Heidelberg},
booktitle = {Proceedings of the 10th Colloquium on Automata, Languages and Programming},
pages = {712–722},
numpages = {11}
}

@article{DBLP:journals/siamcomp/DahlhausJPSY94-complexity-of-multicut,
  author       = {Elias Dahlhaus and
                  David S. Johnson and
                  Christos H. Papadimitriou and
                  Paul D. Seymour and
                  Mihalis Yannakakis},
  title        = {The Complexity of Multiterminal Cuts},
  journal      = {{SIAM} J. Comput.},
  volume       = {23},
  number       = {4},
  pages        = {864--894},
  year         = {1994},
  url          = {https://doi.org/10.1137/S0097539792225297},
  doi          = {10.1137/S0097539792225297},
  bibsource    = {dblp computer science bibliography, https://dblp.org}
}

@article{DBLP:journals/siamcomp/GargVY96-logk-approxi,
  author       = {Naveen Garg and
                  Vijay V. Vazirani and
                  Mihalis Yannakakis},
  title        = {Approximate Max-Flow Min-(Multi)Cut Theorems and Their Applications},
  journal      = {{SIAM} J. Comput.},
  volume       = {25},
  number       = {2},
  pages        = {235--251},
  year         = {1996},
  url          = {https://doi.org/10.1137/S0097539793243016},
  doi          = {10.1137/S0097539793243016},
  bibsource    = {dblp computer science bibliography, https://dblp.org}
}

@article{DBLP:journals/cc/ChawlaKKRS06-approx-harness,
  author       = {Shuchi Chawla and
                  Robert Krauthgamer and
                  Ravi Kumar and
                  Yuval Rabani and
                  D. Sivakumar},
  title        = {On the Hardness of Approximating Multicut and Sparsest-Cut},
  journal      = {Comput. Complex.},
  volume       = {15},
  number       = {2},
  pages        = {94--114},
  year         = {2006},
  url          = {https://doi.org/10.1007/s00037-006-0210-9},
  doi          = {10.1007/S00037-006-0210-9},
  bibsource    = {dblp computer science bibliography, https://dblp.org}
}

@article{DBLP:journals/tcs/Marx06-FPT-open-1,
  author       = {D{\'{a}}niel Marx},
  title        = {Parameterized graph separation problems},
  journal      = {Theor. Comput. Sci.},
  volume       = {351},
  number       = {3},
  pages        = {394--406},
  year         = {2006},
  url          = {https://doi.org/10.1016/j.tcs.2005.10.007},
  doi          = {10.1016/J.TCS.2005.10.007},
  bibsource    = {dblp computer science bibliography, https://dblp.org}
}

@article{DBLP:journals/algorithmica/ChenLL09-FPT-open-3,
  author       = {Jianer Chen and
                  Yang Liu and
                  Songjian Lu},
  title        = {An Improved Parameterized Algorithm for the Minimum Node Multiway
                  Cut Problem},
  journal      = {Algorithmica},
  volume       = {55},
  number       = {1},
  pages        = {1--13},
  year         = {2009},
  url          = {https://doi.org/10.1007/s00453-007-9130-6},
  doi          = {10.1007/S00453-007-9130-6},
  bibsource    = {dblp computer science bibliography, https://dblp.org}
}

@article{DBLP:journals/siamcomp/MarxR14-first-FPT,
  author       = {D{\'{a}}niel Marx and
                  Igor Razgon},
  title        = {Fixed-Parameter Tractability of Multicut Parameterized by the Size
                  of the Cutset},
  journal      = {{SIAM} J. Comput.},
  volume       = {43},
  number       = {2},
  pages        = {355--388},
  year         = {2014},
  url          = {https://doi.org/10.1137/110855247},
  doi          = {10.1137/110855247},
  bibsource    = {dblp computer science bibliography, https://dblp.org}
}

@article{DBLP:journals/siamcomp/BousquetDT18-first-FPT-ind,
  author       = {Nicolas Bousquet and
                  Jean Daligault and
                  St{\'{e}}phan Thomass{\'{e}}},
  title        = {Multicut Is {FPT}},
  journal      = {{SIAM} J. Comput.},
  volume       = {47},
  number       = {1},
  pages        = {166--207},
  year         = {2018},
  url          = {https://doi.org/10.1137/140961808},
  doi          = {10.1137/140961808},
  bibsource    = {dblp computer science bibliography, https://dblp.org}
}

@article{DBLP:journals/tcs/BodlaenderFHMPR10-app-fuzzy-1,
  author       = {Hans L. Bodlaender and
                  Michael R. Fellows and
                  Pinar Heggernes and
                  Federico Mancini and
                  Charis Papadopoulos and
                  Frances A. Rosamond},
  title        = {Clustering with partial information},
  journal      = {Theor. Comput. Sci.},
  volume       = {411},
  number       = {7-9},
  pages        = {1202--1211},
  year         = {2010},
  url          = {https://doi.org/10.1016/j.tcs.2009.12.016},
  doi          = {10.1016/J.TCS.2009.12.016},
  bibsource    = {dblp computer science bibliography, https://dblp.org}
}

@article{DBLP:journals/tcs/DemaineEFI06-app-fuzzy-2,
  author       = {Erik D. Demaine and
                  Dotan Emanuel and
                  Amos Fiat and
                  Nicole Immorlica},
  title        = {Correlation clustering in general weighted graphs},
  journal      = {Theor. Comput. Sci.},
  volume       = {361},
  number       = {2-3},
  pages        = {172--187},
  year         = {2006},
  url          = {https://doi.org/10.1016/j.tcs.2006.05.008},
  doi          = {10.1016/J.TCS.2006.05.008},
  bibsource    = {dblp computer science bibliography, https://dblp.org}
}

@article{DBLP:journals/ml/BansalBC04-app-fuzzy-3,
  author       = {Nikhil Bansal and
                  Avrim Blum and
                  Shuchi Chawla},
  title        = {Correlation Clustering},
  journal      = {Mach. Learn.},
  volume       = {56},
  number       = {1-3},
  pages        = {89--113},
  year         = {2004},
  url          = {https://doi.org/10.1023/B:MACH.0000033116.57574.95},
  doi          = {10.1023/B:MACH.0000033116.57574.95},
  bibsource    = {dblp computer science bibliography, https://dblp.org}
}

@article{DBLP:journals/ipl/MarxR09-approx-fpt,
  author       = {D{\'{a}}niel Marx and
                  Igor Razgon},
  title        = {Constant ratio fixed-parameter approximation of the edge multicut
                  problem},
  journal      = {Inf. Process. Lett.},
  volume       = {109},
  number       = {20},
  pages        = {1161--1166},
  year         = {2009},
  url          = {https://doi.org/10.1016/j.ipl.2009.07.016},
  doi          = {10.1016/J.IPL.2009.07.016},
  bibsource    = {dblp computer science bibliography, https://dblp.org}
}

@article{DBLP:journals/siamcomp/ChitnisHM13-FPT-dir-multiway-cut,
  author       = {Rajesh Hemant Chitnis and
                  MohammadTaghi Hajiaghayi and
                  D{\'{a}}niel Marx},
  title        = {Fixed-Parameter Tractability of Directed Multiway Cut Parameterized
                  by the Size of the Cutset},
  journal      = {{SIAM} J. Comput.},
  volume       = {42},
  number       = {4},
  pages        = {1674--1696},
  year         = {2013},
  url          = {https://doi.org/10.1137/12086217X},
  doi          = {10.1137/12086217X},
  bibsource    = {dblp computer science bibliography, https://dblp.org}
}

@article{DBLP:journals/toct/PilipczukW18a-dir-mul-four,
  author       = {Marcin Pilipczuk and
                  Magnus Wahlstr{\"{o}}m},
  title        = {Directed Multicut is W[1]-hard, Even for Four Terminal Pairs},
  journal      = {{ACM} Trans. Comput. Theory},
  volume       = {10},
  number       = {3},
  pages        = {13:1--13:18},
  year         = {2018},
  url          = {https://doi.org/10.1145/3201775},
  doi          = {10.1145/3201775},
  bibsource    = {dblp computer science bibliography, https://dblp.org}
}

@article{DBLP:journals/talg/ChitnisCHM15-FPT-DSFVS,
  author       = {Rajesh Hemant Chitnis and
                  Marek Cygan and
                  Mohammad Taghi Hajiaghayi and
                  D{\'{a}}niel Marx},
  title        = {Directed Subset Feedback Vertex Set Is Fixed-Parameter Tractable},
  journal      = {{ACM} Trans. Algorithms},
  volume       = {11},
  number       = {4},
  pages        = {28:1--28:28},
  year         = {2015},
  url          = {https://doi.org/10.1145/2700209},
  doi          = {10.1145/2700209},
  bibsource    = {dblp computer science bibliography, https://dblp.org}
}

@book{DBLP:books/sp/CyganFKLMPPS15-book-of-para,
  author       = {Marek Cygan and
                  Fedor V. Fomin and
                  Lukasz Kowalik and
                  Daniel Lokshtanov and
                  D{\'{a}}niel Marx and
                  Marcin Pilipczuk and
                  Michal Pilipczuk and
                  Saket Saurabh},
  title        = {Parameterized Algorithms},
  publisher    = {Springer},
  year         = {2015},
  url          = {https://doi.org/10.1007/978-3-319-21275-3},
  doi          = {10.1007/978-3-319-21275-3},
  isbn         = {978-3-319-21274-6},
  bibsource    = {dblp computer science bibliography, https://dblp.org}
}

@inproceedings{DBLP:conf/stacs/BousquetDTY09-emc-tree-kernel,
  author       = {Nicolas Bousquet and
                  Jean Daligault and
                  St{\'{e}}phan Thomass{\'{e}} and
                  Anders Yeo},
  editor       = {Susanne Albers and
                  Jean{-}Yves Marion},
  title        = {A Polynomial Kernel for Multicut in Trees},
  booktitle    = {26th International Symposium on Theoretical Aspects of Computer Science,
                  {STACS} 2009, February 26-28, 2009, Freiburg, Germany, Proceedings},
  series       = {LIPIcs},
  volume       = {3},
  pages        = {183--194},
  publisher    = {Schloss Dagstuhl - Leibniz-Zentrum f{\"{u}}r Informatik, Germany},
  year         = {2009},
  url          = {https://doi.org/10.4230/LIPIcs.STACS.2009.1824},
  doi          = {10.4230/LIPICS.STACS.2009.1824},
  bibsource    = {dblp computer science bibliography, https://dblp.org}
}

@article{DBLP:journals/networks/GuoN05-emc-tree-fpt,
  author       = {Jiong Guo and
                  Rolf Niedermeier},
  title        = {Fixed-parameter tractability and data reduction for multicut in trees},
  journal      = {Networks},
  volume       = {46},
  number       = {3},
  pages        = {124--135},
  year         = {2005},
  url          = {https://doi.org/10.1002/net.20081},
  doi          = {10.1002/NET.20081},
  bibsource    = {dblp computer science bibliography, https://dblp.org}
}

@article{DBLP:journals/algorithmica/GargVY97-tree-np-hard,
  author       = {Naveen Garg and
                  Vijay V. Vazirani and
                  Mihalis Yannakakis},
  title        = {Primal-Dual Approximation Algorithms for Integral Flow and Multicut
                  in Trees},
  journal      = {Algorithmica},
  volume       = {18},
  number       = {1},
  pages        = {3--20},
  year         = {1997},
  url          = {https://doi.org/10.1007/BF02523685},
  doi          = {10.1007/BF02523685},
  bibsource    = {dblp computer science bibliography, https://dblp.org}
}

@article{DBLP:journals/ipl/GottlobL07-multi-logical,
  author       = {Georg Gottlob and
                  Stephanie Tien Lee},
  title        = {A logical approach to multicut problems},
  journal      = {Inf. Process. Lett.},
  volume       = {103},
  number       = {4},
  pages        = {136--141},
  year         = {2007},
  url          = {https://doi.org/10.1016/j.ipl.2007.03.005},
  doi          = {10.1016/J.IPL.2007.03.005},
  bibsource    = {dblp computer science bibliography, https://dblp.org}
}

@inproceedings{DBLP:conf/ciac/PichlerRW10-multi-tree-decom,
  author       = {Reinhard Pichler and
                  Stefan R{\"{u}}mmele and
                  Stefan Woltran},
  editor       = {Tiziana Calamoneri and
                  Josep D{\'{\i}}az},
  title        = {Multicut Algorithms via Tree Decompositions},
  booktitle    = {Algorithms and Complexity, 7th International Conference, {CIAC} 2010,
                  Rome, Italy, May 26-28, 2010. Proceedings},
  series       = {Lecture Notes in Computer Science},
  volume       = {6078},
  pages        = {167--179},
  publisher    = {Springer},
  year         = {2010},
  url          = {https://doi.org/10.1007/978-3-642-13073-1\_16},
  doi          = {10.1007/978-3-642-13073-1\_16},
  bibsource    = {dblp computer science bibliography, https://dblp.org}
}

@article{DBLP:journals/jal/CalinescuFR03-bounded-deg-tw,
  author       = {Gruia C{u{a}}linescu and
                  Cristina G. Fernandes and
                  Bruce A. Reed},
  title        = {Multicuts in unweighted graphs and digraphs with bounded degree and bounded tree-width},
  journal      = {J. Algorithms},
  volume       = {48},
  number       = {2},
  pages        = {333--359},
  year         = {2003},
  url          = {https://doi.org/10.1016/S0196-6774(03)00073-7},
  doi          = {10.1016/S0196-6774(03)00073-7},
  bibsource    = {dblp computer science bibliography, https://dblp.org}
}

@article{DBLP:journals/jacm/ChuzhoyK09-directed-inapprox,
  author       = {Julia Chuzhoy and
                  Sanjeev Khanna},
  title        = {Polynomial flow-cut gaps and hardness of directed cut problems},
  journal      = {J. {ACM}},
  volume       = {56},
  number       = {2},
  pages        = {6:1--6:28},
  year         = {2009},
  url          = {https://doi.org/10.1145/1502793.1502795},
  doi          = {10.1145/1502793.1502795},
  bibsource    = {dblp computer science bibliography, https://dblp.org}
}

@article{DBLP:journals/eor/GuoHKNU08-graph-class,
  author       = {Jiong Guo and
                  Falk H{\"{u}}ffner and
                  Erhan Kenar and
                  Rolf Niedermeier and
                  Johannes Uhlmann},
  title        = {Complexity and exact algorithms for vertex multicut in interval and
                  bounded treewidth graphs},
  journal      = {Eur. J. Oper. Res.},
  volume       = {186},
  number       = {2},
  pages        = {542--553},
  year         = {2008},
  url          = {https://doi.org/10.1016/j.ejor.2007.02.014},
  doi          = {10.1016/J.EJOR.2007.02.014},
  bibsource    = {dblp computer science bibliography, https://dblp.org}
}

@article{DBLP:journals/mst/Xiao10,
  author       = {Mingyu Xiao},
  title        = {Simple and Improved Parameterized Algorithms for Multiterminal Cuts},
  journal      = {Theory Comput. Syst.},
  volume       = {46},
  number       = {4},
  pages        = {723--736},
  year         = {2010},
  url          = {https://doi.org/10.1007/s00224-009-9215-5},
  doi          = {10.1007/S00224-009-9215-5},
  bibsource    = {dblp computer science bibliography, https://dblp.org}
}

@inproceedings{DBLP:conf/soda/HatzelJLMPSS23,
  author       = {Meike Hatzel and
                  Lars Jaffke and
                  Paloma T. Lima and
                  Tom{\'{a}}s Masar{\'{\i}}k and
                  Marcin Pilipczuk and
                  Roohani Sharma and
                  Manuel Sorge},
  editor       = {Nikhil Bansal and
                  Viswanath Nagarajan},
  title        = {Fixed-parameter tractability of {DIRECTED} {MULTICUT} with three terminal
                  pairs parameterized by the size of the cutset: twin-width meets flow-augmentation},
  booktitle    = {Proceedings of the 2023 {ACM-SIAM} Symposium on Discrete Algorithms,
                  {SODA} 2023, Florence, Italy, January 22-25, 2023},
  pages        = {3229--3244},
  publisher    = {{SIAM}},
  year         = {2023},
  url          = {https://doi.org/10.1137/1.9781611977554.ch123},
  doi          = {10.1137/1.9781611977554.CH123}
}

@article{DBLP:journals/mor/KargerKSTY04,
  author       = {David R. Karger and
                  Philip N. Klein and
                  Clifford Stein and
                  Mikkel Thorup and
                  Neal E. Young},
  title        = {Rounding Algorithms for a Geometric Embedding of Minimum Multiway
                  Cut},
  journal      = {Math. Oper. Res.},
  volume       = {29},
  number       = {3},
  pages        = {436--461},
  year         = {2004},
  url          = {https://doi.org/10.1287/moor.1030.0086},
  doi          = {10.1287/MOOR.1030.0086},
  bibsource    = {dblp computer science bibliography, https://dblp.org}
}

@article{DBLP:journals/jal/GargVY04,
  author       = {Naveen Garg and
                  Vijay V. Vazirani and
                  Mihalis Yannakakis},
  title        = {Multiway cuts in node weighted graphs},
  journal      = {J. Algorithms},
  volume       = {50},
  number       = {1},
  pages        = {49--61},
  year         = {2004},
  url          = {https://doi.org/10.1016/S0196-6774(03)00111-1},
  doi          = {10.1016/S0196-6774(03)00111-1},
  bibsource    = {dblp computer science bibliography, https://dblp.org}
}

@article{DBLP:journals/siamcomp/NaorZ01,
  author       = {Joseph Naor and
                  Leonid Zosin},
  title        = {A 2-Approximation Algorithm for the Directed Multiway Cut Problem},
  journal      = {{SIAM} J. Comput.},
  volume       = {31},
  number       = {2},
  pages        = {477--482},
  year         = {2001},
  url          = {https://doi.org/10.1137/S009753979732147X},
  doi          = {10.1137/S009753979732147X},
  bibsource    = {dblp computer science bibliography, https://dblp.org}
}

@article{DBLP:journals/ipl/CaoCF14,
  author       = {Yixin Cao and
                  Jianer Chen and
                  Jia{-}Hao Fan},
  title        = {An O(1.84\({}^{\mbox{k}}\)) parameterized algorithm for the multiterminal
                  cut problem},
  journal      = {Inf. Process. Lett.},
  volume       = {114},
  number       = {4},
  pages        = {167--173},
  year         = {2014},
  url          = {https://doi.org/10.1016/j.ipl.2013.12.001},
  doi          = {10.1016/J.IPL.2013.12.001},
  bibsource    = {dblp computer science bibliography, https://dblp.org}
}

@inproceedings{DBLP:conf/focs/NaorSS95,
  author       = {Moni Naor and
                  Leonard J. Schulman and
                  Aravind Srinivasan},
  title        = {Splitters and Near-Optimal Derandomization},
  booktitle    = {36th Annual Symposium on Foundations of Computer Science, {FOCS} 1995,
                  Milwaukee, Wisconsin, USA, 23-25 October 1995},
  pages        = {182--191},
  publisher    = {{IEEE} Computer Society},
  year         = {1995},
  url          = {https://doi.org/10.1109/SFCS.1995.492475},
  doi          = {10.1109/SFCS.1995.492475},
  bibsource    = {dblp computer science bibliography, https://dblp.org}
}

@book{downey2013fundamentals,
  title={Fundamentals of parameterized complexity},
  author={Downey, Rodney G and Fellows, Michael R and others},
  volume={4},
  year={2013},
  publisher={Springer}
}

@article{DBLP:journals/jacm/ChenLLOR08,
  author       = {Jianer Chen and
                  Yang Liu and
                  Songjian Lu and
                  Barry O'Sullivan and
                  Igor Razgon},
  title        = {A fixed-parameter algorithm for the directed feedback vertex set problem},
  journal      = {J. {ACM}},
  volume       = {55},
  number       = {5},
  pages        = {21:1--21:19},
  year         = {2008}
}

@article{DBLP:journals/jcss/RazgonO09,
  author       = {Igor Razgon and
                  Barry O'Sullivan},
  title        = {Almost 2-SAT is fixed-parameter tractable},
  journal      = {J. Comput. Syst. Sci.},
  volume       = {75},
  number       = {8},
  pages        = {435--450},
  year         = {2009}
}

@article{DBLP:journals/toct/CyganPPW13,
  author       = {Marek Cygan and
                  Marcin Pilipczuk and
                  Michal Pilipczuk and
                  Jakub Onufry Wojtaszczyk},
  title        = {On multiway cut parameterized above lower bounds},
  journal      = {{ACM} Trans. Comput. Theory},
  volume       = {5},
  number       = {1},
  pages        = {3:1--3:11},
  year         = {2013}
}

@inproceedings{DBLP:conf/soda/LokshtanovR0Z20,
  author       = {Daniel Lokshtanov and
                  M. S. Ramanujan and
                  Saket Saurabh and
                  Meirav Zehavi},
  editor       = {Shuchi Chawla},
  title        = {Parameterized Complexity and Approximability of Directed Odd Cycle
                  Transversal},
  booktitle    = {Proceedings of the 2020 {ACM-SIAM} Symposium on Discrete Algorithms,
                  {SODA} 2020, Salt Lake City, UT, USA, January 5-8, 2020},
  pages        = {2181--2200},
  publisher    = {{SIAM}},
  year         = {2020}
}

@article{DBLP:journals/iandc/LokshtanovM13,
  author       = {Daniel Lokshtanov and
                  D{\'{a}}niel Marx},
  title        = {Clustering with local restrictions},
  journal      = {Inf. Comput.},
  volume       = {222},
  pages        = {278--292},
  year         = {2013}
}

@article{DBLP:journals/siamdm/KratschPPW15,
  author       = {Stefan Kratsch and
                  Marcin Pilipczuk and
                  Michal Pilipczuk and
                  Magnus Wahlstr{\"{o}}m},
  title        = {Fixed-Parameter Tractability of Multicut in Directed Acyclic Graphs},
  journal      = {{SIAM} J. Discret. Math.},
  volume       = {29},
  number       = {1},
  pages        = {122--144},
  year         = {2015}
}

@inproceedings{DBLP:conf/soda/LokshtanovMRSZ21,
  author       = {Daniel Lokshtanov and
                  Pranabendu Misra and
                  M. S. Ramanujan and
                  Saket Saurabh and
                  Meirav Zehavi},
  editor       = {D{\'{a}}niel Marx},
  title        = {FPT-approximation for {FPT} Problems},
  booktitle    = {Proceedings of the 2021 {ACM-SIAM} Symposium on Discrete Algorithms,
                  {SODA} 2021, Virtual Conference, January 10 - 13, 2021},
  pages        = {199--218},
  publisher    = {{SIAM}},
  year         = {2021}
}

@inproceedings{DBLP:conf/stoc/KorhonenL23,
  author       = {Tuukka Korhonen and
                  Daniel Lokshtanov},
  editor       = {Barna Saha and
                  Rocco A. Servedio},
  title        = {An Improved Parameterized Algorithm for Treewidth},
  booktitle    = {Proceedings of the 55th Annual {ACM} Symposium on Theory of Computing,
                  {STOC} 2023, Orlando, FL, USA, June 20-23, 2023},
  pages        = {528--541},
  publisher    = {{ACM}},
  year         = {2023}
}

@article{DBLP:journals/siamdm/KimMPSW24,
  author       = {Eun Jung Kim and
                  Tom{\'{a}}s Masar{\'{\i}}k and
                  Marcin Pilipczuk and
                  Roohani Sharma and
                  Magnus Wahlstr{\"{o}}m},
  title        = {On Weighted Graph Separation Problems and Flow Augmentation},
  journal      = {{SIAM} J. Discret. Math.},
  volume       = {38},
  number       = {1},
  pages        = {170--189},
  year         = {2024},
  url          = {https://doi.org/10.1137/22m153118x},
  doi          = {10.1137/22M153118X},
  bibsource    = {dblp computer science bibliography, https://dblp.org}
}

@article{DBLP:journals/jacm/KimKPW25,
  author       = {Eun Jung Kim and
                  Stefan Kratsch and
                  Marcin Pilipczuk and
                  Magnus Wahlstr{\"{o}}m},
  title        = {Flow-augmentation {I:} Directed graphs},
  journal      = {J. {ACM}},
  volume       = {72},
  number       = {1},
  pages        = {5:1--5:38},
  year         = {2025},
  url          = {https://doi.org/10.1145/3706103},
  doi          = {10.1145/3706103},
  bibsource    = {dblp computer science bibliography, https://dblp.org}
}

@inproceedings{DBLP:conf/iwpec/Guillemot08a,
  author       = {Sylvain Guillemot},
  editor       = {Martin Grohe and
                  Rolf Niedermeier},
  title        = {{FPT} Algorithms for Path-Transversals and Cycle-Transversals Problems
                  in Graphs},
  booktitle    = {Parameterized and Exact Computation, Third International Workshop,
                  {IWPEC} 2008, Victoria, Canada, May 14-16, 2008. Proceedings},
  series       = {Lecture Notes in Computer Science},
  volume       = {5018},
  pages        = {129--140},
  publisher    = {Springer},
  year         = {2008},
  url          = {https://doi.org/10.1007/978-3-540-79723-4\_13},
  doi          = {10.1007/978-3-540-79723-4\_13},
  bibsource    = {dblp computer science bibliography, https://dblp.org}
}

@book{vazirani2001approximation,
  title={Approximation algorithms},
  author={Vazirani, Vijay V},
  volume={1},
  year={2001},
  publisher={Springer}
}
\bibliographystyle{alpha}

\end{document}